\documentclass[11pt]{amsart}
\newcommand{\ignore}[1]{}
\usepackage{a4wide}
\usepackage{lmodern}
\usepackage{stix}
\usepackage{amssymb}
\usepackage{amsfonts}
\usepackage{microtype}
\usepackage{amsmath}
\usepackage{booktabs} 
\usepackage{amssymb}
\usepackage{amssymb}
\usepackage{amsthm}
\usepackage{bbm} 
\usepackage{textcomp}
\usepackage{xcolor}
\usepackage{graphicx}

\usepackage[colorlinks]{hyperref}
\definecolor{gruen}{rgb}{0,.5,0}
\hypersetup{colorlinks=true,linkcolor=gruen,citecolor=gruen,urlcolor=gruen}

\usepackage[mathcal,mathscr]{euscript}
\usepackage{caption} 
\usepackage[rightcaption]{sidecap}

\usepackage[capitalise]{cleveref}
\crefname{lem}{Lemma}{Lemmas}
\crefname{thm}{Theorem}{Theorems}
\crefname{result}{Result}{Results}
\crefname{prop}{Proposition}{Propositions}

\usepackage{microtype} 

\newtheorem{thm}{Theorem}[section]
\newtheorem{lem}[thm]{Lemma}
\newtheorem{cor}{Corollary}[section]
\newtheorem{prop}[thm]{Proposition}
\newtheorem*{clm}{Claim}
\newtheorem{problem}{Problem}

\theoremstyle{remark}
\newtheorem{rem}{Remark}
\newtheorem{ex}{Example}
\newtheorem{result}{Result}

\let\geq\geqslant
\let\leq\leqslant
\let\bar\underline
\let\cal\mathcal
\newcommand{\NN}{\mathbb{N}}

\newcommand{\RR}{\mathbb{R}}

\newcommand{\f}{{\cal F}}
\newcommand{\h}{{\cal H}}
\newcommand{\A}{{\cal A}}
\newcommand{\B}{{\cal B}}
\newcommand{\R}{{\cal R}}
\newcommand{\skal}[1]{\langle #1\rangle}
\newcommand{\Un}[1]{L(#1)}
\newcommand{\Min}[2]{\mathsf{Min}_{#2}(#1)}
\newcommand{\Max}[2]{\mathsf{Max}_{#2}(#1)}
\newcommand{\Bool}[1]{\mathsf{Bool}(#1)}
\newcommand{\BBool}[2]{\mathsf{Bool}_{#2}(#1)}
\newcommand{\gf}[1]{\mathrm{GF}(#1)}
\newcommand{\supp}[1]{S_{\!#1}}
\def\deg#1{\mathrm{deg}(#1)}
\newcommand{\ddeg}[2]{\#_{#2}(#1)}
\newcommand{\err}{\epsilon}
\newcommand{\bal}{\gamma} 
\newcommand{\bbal}{\theta}
\newcommand{\conv}[1]{\mathrm{Conv}(#1)}
\newcommand{\compl}[1]{Y_{#1}}
\newcommand{\pr}[1]{X_{#1}}
\newcommand{\BF}[2]{u^{#1}_{#2}}
\newcommand{\dist}[1]{\mathrm{dist}(#1)} 
\newcommand{\spanntree}{\mathcal{T}}
\newcommand{\matching}{\mathcal{M}}
\newcommand{\euler}{\mathrm{e}}
\newcommand{\Nor}{\mu} 
\newcommand{\nor}[1]{\Nor(#1)}
\newcommand{\nnor}[2]{\Nor_{#2}(#1)}
\newcommand{\mm}{p} %
\newcommand{\vnul}{\vec{0}}

\newcommand{\vienas}{\mathsf{1}}
\newcommand{\const}[1]{\lambda_{#1}}

\newcommand{\scalar}{c}
\newcommand{\dff}{d}
\newcommand{\xx}{x_0}
\newcommand{\zz}{z}
\newcommand{\xf}{\boldsymbol{\mathcal{F}}}
\newcommand*{\daug}{\mathbin{\scalebox{0.9}{\ensuremath{\otimes}}}}%
\newcommand*{\suma}{\mathbin{\scalebox{0.9}{\ensuremath{\oplus}}}}%
\newcommand{\vvec}[2]{\vec{#1}_{#2}}
\newcommand{\Vvec}[2]{\vec{#1}_{#2}'}
\newcommand{\bool}[1]{f_{#1}} 
\newcommand{\bphi}{\phi} 
\newcommand{\sel}[2]{f_{#1,#2}} 
\newcommand{\fgreedy}[1]{r_{\mathrm{greed}}(#1)}
\newcommand{\case}[1]{\medskip \item {\it Case} #1:}

\usepackage[all,arc,curve,color,frame]{xy}

\def\bild{
\begin{figure}
  \begin{center}
    $ \scalebox{0.8}{ \xygraph { !{(0,0)}*+{x}="x" !{(1.6,0)}*+{y}="y"
        !{(0,-0.7)}*+=[o]+[F]{+}="a" !{(1.6,-0.7)}*+=[o]+[F]{+}="b"
        !{(0.8,-1.5)}*+=[o]+[F]{\min}="c" "x":@{=>}"a" "y":@{=>}"b"
        "a":"c" "b":"c" } } $ \qquad \qquad $ \scalebox{0.8}{ \xygraph
      { 
        !{(0,0)}*+{x}="x" !{(1.6,0)}*+{y}="y"
        !{(0.8,-0.7)}*+=[o]+[F]{\min}="a"
        !{(0.8,-1.5)}*+=[o]+[F]{+}="b" "x":"a" "y":"a" "a":@{=>}"b"
      } } $
  \end{center}
  \caption{Two constant-free $(\min,+)$ circuits solving the
    minimization problem $f(x,y)=\min\{2x,2y\}$ whose set of feasible
    solutions is $A=\{(2,0), (0,2)\}$. The first circuit produces the
    set $A$ itself, whereas the second saves one gate by producing a
    \emph{different} set $B=\{(2,0), (1,1), (0,2)\}$. Here
    $\Downarrow$ stands for two parallel edges.}
  \label{fig:produced}
\end{figure}
}

\def\floyd{
\begin{SCfigure}[10]
  $\scalebox{0.8}{ \xygraph { 
      !{(-1.3,0)}*+{g_{k-1}(i,j)}="a"
      !{(0.5,0)}*+{g_{k-1}(i,k)}="b"
      !{(2.5,0)}*+{g_{k-1}(k,j)}="c"
      !{(0.5,-1.4)}*+=[o]+[F]{\min}="g1"
      !{(1.5,-0.7)}*+=[o]+[F]{+}="g2"
      !{(0.5,-2)}*+={g_{k}(i,j)}="d"
      "a":"g1" "b":"g2" "c":"g2" "g2":"g1" } } $
  \caption[.]{\small A fragment of a tropical $(\min,+)$ circuit of
    size $O(n^3)$ implementing the Floyd--Warshall DP algorithm for
    the all-pairs lightest paths on $K_n$ problem.  At the gate
    $g_k(i,j)$, the minimum weight of a path from $i$ to $j$, which
    only uses nodes $1,\ldots,k$ as inner nodes, is computed.  }
  \label{fig:floyd}
\end{SCfigure}
}

\begin{document}

\title[]{Approximation Limitations of Pure Dynamic Programming}
\thanks{This work was funded by the DFG (German
      Research Foundation)  grant~JU~3105/1-1.}
      
\author[]{Stasys Jukna$^1$ \and Hannes Seiwert$^2$}
\thanks{$^1$Faculty of Mathematics and Computer Science, Vilnius
    University, Lithuania. Email: stjukna@gmail.com, homepage: \href{http://www.thi.cs.uni-frankfurt.de/~jukna/}
    {http://www.thi.cs.uni-frankfurt.de/$\sim$jukna/}}
\thanks{$^2$Institute of Computer
      Science, Goethe University Frankfurt, Frankfurt am Main, Germany. Email: seiwert@thi.cs.uni-frankfurt.de.}

\begin{abstract}
  We prove the first, even super-polynomial, lower bounds on the size
  of tropical (min,+) and (max,+) circuits approximating given
  optimization problems.  Many classical dynamic programming (DP)
  algorithms for  optimization problems are pure in that
  they only use the basic $\min$, $\max$, $+$ operations in their recursion
  equations.  Tropical circuits constitute a rigorous mathematical
  model for this class of algorithms. An algorithmic consequence of
  our lower bounds for tropical circuits is that the approximation
  powers of pure DP algorithms and greedy algorithms are
  incomparable. That pure DP algorithms can hardly beat greedy in
  approximation, is long known. New in this consequence is that also
  the converse holds.
\end{abstract}

\maketitle


\section{Introduction}
A combinatorial optimization problem is specified by a finite set of
\emph{ground elements} and a family $\f$ of subsets of these elements,
called \emph{feasible solutions}. The problem itself then is, given an
assignment of nonnegative real weights to the ground elements, to
compute the minimum or the maximum weight of a feasible solution, the
latter being the sum of weights of its elements.

The family $\f$ of feasible solutions itself can be described either
explicitly, or as the set of $0$-$1$ solutions of a system of linear
inequalities (as in linear programming), or by other means. Important
is only that $\f$ does \emph{not} depend on the actual input
weighting: the family $\f$ is the \emph{same} for \emph{all} arriving
input weightings.

For example, in the MST problem (minimum weight spanning tree problem)
on a given graph, feasible solutions are spanning trees of this graph
(viewed as sets of their edges), and the problem is to compute the
minimum weight of a spanning tree of this graph. In the assignment
problem, feasible solutions are perfect matchings in a complete
bipartite graph, etc.

Dynamic programming (DP) is a fundamental algorithmic paradigm for
solving combinatorial optimization problems.  Many classical DP
algorithms are \emph{pure} in that they only apply the \emph{basic}
operations $(\min,+)$ or $(\max,+)$ in their recursion equations. Note
that these are the \emph{only} operations used in the definitions of
the optimization problems themselves.

Notable examples of pure DP algorithms for combinatorial optimization
problems are the well-known Bell\-man--Ford--Moore shortest $s$-$t$
path algorithm \cite{bellman,ford,Moore1957}, the Floyd--Warshall
all-pairs shortest paths algorithm~\cite{floyd,warshall} (see
\cref{fig:floyd}), the Held--Karp traveling salesman
algorithm~\cite{held62}, the Dreyfus--Levin--Wagner Steiner tree
algorithm~\cite{dreyfus,levin}. The Viterbi $(\max,\times)$ DP
algorithm~\cite{viterbi} is also a pure $(\min,+)$ DP algorithm via
the isomorphism $h:(0,1]\to\RR_+$ given by $h(x)=-\ln x$.

The main question we ask in this paper is: How many operations are
necessary for pure DP algorithms to \emph{approximate} a given
combinatorial optimization problem within a given factor?  That is, we
are interested in proving \emph{lower bounds} on the number of
performed operations.

A natural mathematical model for pure $(\min,+)$ and $(\max,+)$ DP
algorithms is that of tropical circuits. A \emph{tropical $(\min,+)$
  circuit} is a directed acyclic graph, whose each indegree-zero node
holds either one of the input variables $x_1,\ldots,x_n$ or a
nonnegative real constant, and every other node (a gate) has indegree
two and computes either the minimum or the sum of the values computed
at its two predecessors. Tropical $(\max,+)$ circuits are defined
similarly. The \emph{size} of a circuit is the total number of its
gates. Note that pure $(\min,+)$ and $(\max,+)$ DP algorithms are just
special (recursively constructed) tropical circuits (see
\cref{fig:floyd}). So, lower bounds on the size of tropical circuits
show limits of these pure DP algorithms.

In this paper, we prove the first non-trivial, even super-polynomial,
lower bounds for \emph{approximating} tropical circuits and, hence,
also for approximating pure DP algorithms.

Recall that an algorithm \emph{approximates} a given optimization
problem $f$ within a \emph{factor} $r\geq 1$ (or
$r$-\emph{approximates}~$f$) if for every input weighting $x$ (a
vector of $n$ nonnegative real numbers), the output value of the
algorithm lies:
\begin{itemize}
\item[$\circ$] between $f(x)$ and $r\cdot f(x)$, in the case when $f$
  is a \emph{minimization} problem;

\item[$\circ$] between $f(x)/r$ and $f(x)$, in the case when $f$ is a
  \emph{maximization} problem.
\end{itemize}
The factor $r$ may depend on the length $n$ of the inputs $x$, but not
on the inputs $x$ themselves. In both cases, the \emph{smaller} the
factor $r$ is, the better is the approximation. In particular, factor
$r=1$ means that the problem is solved exactly.

\floyd

One of our motivations for proving lower bounds on the number of
operations performed by \emph{approximating} pure DP algorithms is to
compare their approximation power with that of the greedy algorithm;
see \cref{app:greedy} for what we mean by \emph{the} greedy algorithm.

That the greedy algorithm can have much worse approximation behavior
than pure DP algorithms is long known. Namely, there are many
optimization problems easily solvable by pure DP algorithms using a
small number of $(\min,+)$ or $(\max,+)$ operations, but the greedy
algorithm cannot achieve any non-trivial approximation factor (smaller
than the maximum number of elements in feasible solutions). Such are,
for example, the maximum weight independent set in a tree, or the
maximum weight simple path in a transitive tournament problem, and
many other problems. To give a trivial example, note that the problem
$f(x)=\max\{x_1,x_2+\cdots+x_n\}$ can be solved (within factor $r=1$)
by a trivial pure $(\max,+)$ DP algorithm performing only $n-1$
operations, but the greedy algorithm cannot achieve any smaller than
$r=n-1$ approximation factor for this problem (see \cref{prop:greedy}
in \cref{app:greedy}).  But what about the \emph{converse} direction:
can also pure DP algorithms have worse approximation behavior than
greedy?

Apparently, the first indication that greedy can also beat pure DP was
given by Jerrum and Snir~\cite{jerrum}. They proved that every
$(\min,+)$ circuit solving (exactly, within factor $r=1$) the
\emph{directed} MST problem on $n$-vertex graphs (known also as the
\emph{arborescence} problem) requires $2^{\Omega(n)}$ gates. Since the
family of feasible solutions of the arborescence problem is an
intersection of two matroids, the greedy algorithm can approximate
this problem within factor $r=2$. This result was later improved in
\cite{JS19} by showing that also the \emph{undirected} MST problem,
which can already be solved by the greedy algorithm exactly, requires
$(\min,+)$ circuits of size $2^{\Omega(\sqrt{n})}$ to be solved
exactly.

But what if pure DP algorithms are only required to \emph{approximate}
a given optimization problem within some factor $r > 1$?  Can greedy
algorithms achieve smaller approximation factors than efficient pure
DP algorithms? Our lower bounds on the size of approximating tropical
circuits answer this question in the \emph{affirmative}.

Below we summarize our main results. Since the approximation behaviors
of tropical $(\min,+)$ and $(\max,+)$ circuits turned out to be
completely different, we consider minimization and maximization
problems separately.

\section{Main results}
\label{sec:results}

Recall that a combinatorial optimization problem $f(x_1,\ldots,x_n)$
is specified by giving some family $\f\subseteq 2^{[n]}$ of feasible
solutions. The problem itself is then, given an input weighting
$x\in\RR_+^n$, to compute either the minimum or the maximum weight
$\sum_{i\in S}x_i$ of a feasible solution $S\in\f$.  To indicate the
total number $n$ of ground elements, we will also write $f_n$ instead
of just $f$.

\subsection*{Minimization}
The \emph{boolean version} of a minimization problem $f_n$ is the
monotone boolean function which, given a set of ground elements,
decides whether this set contains at least one feasible solution
of~$f_n$.

\begin{result}[Boolean bound for $(\min,+)$ circuits; \cref{thm:bool}]
  \label{res:bool}
  If the boolean version of a minimization problem $f_n$ requires
  monotone boolean $(\lor,\land)$ circuits of size $>t$, then no
  tropical $(\min,+)$ circuit of size $\leq t$ can approximate $f_n$
  within any finite factor $r=r(n)\geq 1$.
\end{result}

That is, if a tropical $(\min,+)$ circuit has fewer than $t$ gates,
then regardless of how large approximation factor $r$ we will allow,
there will be an input weighting on which the circuit makes an error:
the computed value on this input will be either strictly smaller or
more than $r$ times larger than the optimal value.

Together with known lower bounds for monotone boolean circuits,
\cref{res:bool} yields the same lower bounds for tropical $(\min,+)$
circuits approximating the corresponding minimization problems.

Take, for example the assignment problem: given a nonnegative
weighting of the edges of the complete bipartite $n\times n$ graph,
compute the minimum weight of a perfect matching.  Jerrum and
Snir~\cite{jerrum} have proved that any $(\min,+)$ circuit solving
this problem exactly (within the factor $r=1$) must have
$2^{\Omega(n)}$ gates.  On the other hand, together with Razborov's
monotone circuit lower bound for the logical permanent
function~\cite{razb-perm}, \cref{res:bool} implies that a polynomial
in $n$ number of gates is not sufficient to approximate this problem
even when an \emph{arbitrarily} large approximation factor is allowed:
for \emph{any} finite approximation factor $r=r(n)\geq 1$, at least
$n^{\Omega(\log n)}$ gates are necessary to approximate the assignment
problem within the factor~$r$.

By combining the boolean bound (\cref{res:bool}) with counting
arguments, we show that the greedy algorithm can beat approximating
pure $(\min,+)$ DP algorithms on some minimization problems.

\begin{result}[Greedy \mbox{}can beat $(\min,+)$ circuits;
  \cref{thm:min-gap}]
  \label{res:min-gap}
  There are doubly-exponentially many in $n$ minimization problems
  $f_n$ such that the greedy algorithm solves $f_n$ exactly, but any
  $(\min,+)$ circuit approximating $f_n$ within any finite factor
  $r=r(n)\geq 1$ must have $2^{\Omega(n)}$ gates.
\end{result}

Our proof of \cref{res:bool} is fairly simple, but it only gives us an
``absolute'' lower bound on the number of gates, below which no
tropical $(\min,+)$ circuit can approximate a given minimization
problem within \emph{any} factor. More interesting (and less simple),
however, is the fact that, after an appropriate definition of the
``semantic degree'' of monotone boolean circuits (\cref{sec:sem-deg}),
also a \emph{converse} of \cref{res:bool} holds: the approximation
power of tropical $(\min,+)$ circuits is \emph{captured} (not only
lower bounded) by the computational power of monotone boolean circuits
of bounded semantic degree.

\begin{result}[Converse of the boolean bound; \cref{thm:minexact}]
  \label{res:minexact}
  A minimization problem $f$ can be approximated within a factor $r$
  by a tropical $(\min,+)$ circuit of size~$t$ if and only if the
  boolean version of $f$ can be computed by a monotone boolean
  $(\lor,\land)$ circuit of size $t$ and semantic degree at most~$r$.
\end{result}
We prove this result in \cref{sec:bool-tight} using convexity
arguments. Yet another consequence of these arguments is (see
\cref{rem:min-to-arithm}) that, in order to show that the minimization
problem on a family $\f\subseteq 2^{[n]}$ of feasible solutions can be
$r$-approximated by a $(\min,+)$ circuit of size $t$, it is enough to
design a monotone \emph{arithmetic} $(+,\times)$ circuit of size $\leq
t$ such that the  polynomial computed by this circuit has the following
two properties (where we, as customary, only consider monomials with nonzero coefficients):
\begin{enumerate}
\item for every monomial $\prod_{i\in T}x_i^{d_i}$ there is a set
  $S\in\f$ with $S\subseteq T$;
\item for every set $S\in\f$ there is a monomial $\prod_{i\in
    T}x_i^{d_i}$ with $T=S$ and all $d_i\leq r$.
\end{enumerate}
That is, we can approximate minimization problems by designing
monotone arithmetic circuits of bounded degree. This is a (rough)
\emph{upper} bound on the size of approximating $(\min,+)$ circuits in
terms of \emph{arithmetic} circuits. \Cref{res:minexact} gives a
\emph{tight} bound, but in terms of \emph{boolean} circuits.

\subsection*{Maximization}

It turned out that not only the approximation behaviors of $(\min,+)$
and $(\max,+)$ circuits are different (approximation factors may be
unbounded in the former model, while they are always bounded in the
latter model), but also the task of proving lower bounds for
approximating $(\max,+)$ circuits is by far more difficult than that
for $(\min,+)$ circuits.

The point is that for approximating $(\max,+)$ circuits, even Shannon
type \emph{counting} arguments fail (see \cref{sec:count-fails}).  In
particular, there are doubly-exponentially many in $n$ maximization
problems $f_n$ such that $(\max,+)$ circuits require $2^{\Omega(n)}$
gates to solve any of them exactly (within the factor $r=1$), but
\emph{one single} $(\max,+)$ circuit of size $O(n^2)$ approximates
each of these problems within a just slightly larger than~$1$ factor
$r=1+o(1)$ (\cref{prop:counting1}).  Such a jump in circuit size
occurs also on \emph{random} maximization problems
(\cref{prop:counting2}). Moreover, there are also \emph{explicit}
maximization problems $f_n$ that require $(\max,+)$ circuits of size at
least $2^{n/4}$ to solve them exactly (within factor $r=1$), but can
be approximated within the factor $r=2$ by using only $n$ gates
(\cref{thm:sidon}).

Being warned by these facts, we go much deeper (than in the case of
minimization) into the structure of approximating $(\max,+)$ circuits
and prove a general ``rectangle lower bound'' for them.

Let $\f$ be a family of feasible solutions.  A~\emph{rectangle} is a
family of sets specified by a pair $\A,\B$ of families satisfying
$A\cap B=\emptyset$ for all $A\in \A$ and $B\in\B$. The rectangle
$\R=\A\lor \B$ itself consists of all sets $A\cup B$ with $A\in \A$
and $B\in\B$.  The rectangle $\R$ \emph{lies below} $\f$ if every set
of $\R$ is contained in at least one set of $\f$.  Given an
approximation factor $r\geq 1$, we say that a set $F\in\f$
\emph{appears $r$-balanced} in the rectangle $\R$ if there are sets
$A\in\A$ and $B\in\B$ such that $F$ shares $\geq |F|/r$ elements with
$A\cup B$, and $\geq |F|/3r$ elements with both $A$ and~$B$.

\begin{result}[Rectangle bound; special case of \cref{thm:rect}]
  \label{res:rect}
  If in any rectangle lying below $\f$, at most a $1/t$ portion of
  sets of $\f$ appear $r$-balanced, then every $(\max,+)$ circuit
  approximating the maximization problem on $\f$ within the factor $r$
  must have at least $t$ gates.
\end{result}

Using the rectangle bound, we show that already a slight decrease of
the allowed approximation factors $r$ can make tractable problems
intractable, and that this happens for arbitrarily large
factors~$r$. In the following result formalizing this phenomenon,
$\err>0$ is an arbitrarily small constant.

\begin{result}[Factor hierarchy theorem; \cref{thm:hierarchy}]\label{res:hierarchy}
  For every prime power $m$ and integer $1\leq d\leq m$, there is an
  explicit maximization problem $f_n$ on $n=m^2$ ground elements which
  can be approximated within the factor $r=m/d$ by a $(\max,+)$
  circuit of size $3n$, but any $(\max,+)$ circuit approximating $f_n$
  within the factor $(1-\err)r$ must have at least $n^{\err d/4}$
  gates.
\end{result}

Finally, using the rectangle bound, we show that there are explicit
maximization problems $f_n$ such that $(\max,+)$ circuits of
polynomial in $n$ size cannot achieve even an exponentially larger
factor than the factor achieved by the greedy algorithm on~$f_n$.

\begin{result}[Greedy can beat $(\max,+)$ circuits; \cref{thm:matching}]\label{res:matching}
  For every integer $r\geq 6$, there are explicit maximization
  problems $f_n$ such that the greedy algorithm approximates $f_n$
  within the factor $r$, but every $(\max,+)$ circuit approximating
  $f_n$ within the factor $2^r/9$ must have $2^{n^{\Omega(1)}}$ gates.
\end{result}

Families of feasible solutions of the maximization problems $f_n$ in
\cref{res:hierarchy} are particular combinatorial designs, while those
in \cref{res:matching} are families of perfect matchings in
$r$-partite $r$-uniform hypergraphs.

\subsection*{The algorithmic message}
As we already mentioned above, it was long known that for some
combinatorial optimization problems, greedy algorithms can have much
worse approximation behavior than pure DP algorithms. Thus,
\cref{res:min-gap,res:matching} imply that the approximation powers of
greedy and pure DP algorithms are \emph{incomparable}: on some
optimization problems, pure DP algorithms can also have much
\emph{worse} approximation behavior than greedy.

\subsection*{Why ``only'' pure DP?}
In this paper, we only consider pure $(\min,+)$ and $(\max,+)$ DP
algorithms.  Non-pure DP algorithms may use other arithmetic
operations, rounding, as well as very powerful operations like
conditional branchings (via if-then-else constraints), argmin, argmax,
etc. The presence of such operations makes the corresponding circuit
models no longer amenable for analysis using known mathematical
tools. In particular, such DP algorithms have the full power of
arithmetic circuits as well as of unrestricted boolean
$(\lor,\land,\neg)$ circuits (for example, $\neg x$ is a simple
conditional branching operation {\sf if} $x=0$ {\sf then} $1$ {\sf
  else} $0$).  Let us stress that our goal is to prove (unconditional)
\emph{lower bounds}. In the context of this task, even proving lower
bounds for exactly solving $(\min,+,-)$ circuits (tropical circuits
with subtraction operation allowed), remains a challenge (see
\cref{sec:subtraction}).

\subsection*{Organization} In \cref{sec:preliminaries}, we recall the
concept of sets \emph{produced} by circuits, and show that when
approximating combinatorial optimization problems, we can safely
assume that tropical circuits are \emph{constant-free}, that is,
contain no constants as inputs (\cref{lem:constant-free}).
\Cref{sec:minplus,sec:maxplus,sec:structure-tight,sec:bool-tight} are
devoted to the proofs of our main results. \Cref{res:bool,res:min-gap}
are proved in \cref{sec:minplus}, and
\cref{res:rect,res:hierarchy,res:matching} are proved in
\cref{sec:maxplus}.  In \cref{sec:structure-tight}, we use convexity
arguments (Farkas' lemma) to give a \emph{tight} structural connection
between the sets of feasible solutions of optimization problems to be
\emph{approximated} and the sets of feasible solutions \emph{produced}
by approximating tropical circuits. In \cref{sec:bool-tight}, we prove
the \emph{converse} of our boolean lower bound for approximating
$(\min,+)$ circuits (\cref{res:minexact}). The concluding section
(\cref{sec:concl}) contains some open problems. In \cref{app:greedy},
we recall greedy algorithms. In \cref{app:sidon}, we exhibit an
exponential (almost maximal possible) decrease in the size of
$(\max,+)$ circuits on \emph{explicit} maximization problems when
going from the approximation factor $r=1$ (exact solution) to factor
$r=2$.

\subsection*{Notation}
Through the paper, $\NN=\{0,1,2,\ldots\}$ will denote the set of all
nonnegative integers, $[n]=\{1,2,\ldots,n\}$ the set of the first $n$
positive integers, $\RR_+$ the set of all nonnegative real numbers,
and $2^E$ the family of all subsets of a set~$E$. Also, $\vnul$ will
denote the all-$0$ vector, $\vec{e}_i$ will denote the $0$-$1$ vector
with exactly one $1$ in the $i$th position. For sets
$A,B\subseteq\RR^n$ of vectors, their \emph{Minkowski sum} (or
\emph{sumset}) is the set of vectors $A+B=\{a+b\colon a\in A, b\in
B\}\subseteq\RR^n$, where $a+b=(a_1+b_1,\ldots,a_n+b_n)$ is the
componentwise sum of vector $a$ and $b$. That is, we add every vector
of $B$ to every vector of~$A$.  For a real vector $a=(a_1,\ldots,a_n)$
and a scalar $\lambda\in\RR$, $\lambda\cdot a$ stands for the vector
$(\lambda a_1,\ldots,\lambda a_n)$. If $A\subseteq\RR^n$ is a set of
vectors, then $\lambda\cdot A$ stands for the set of vectors
$\{\lambda\cdot a\colon a\in A\}$. The \emph{support} of vector $a$ is
the set $\supp{a}=\{i\colon a_i\neq 0\}$ of its nonzero positions.

As customary, a family $\f$ of sets is an \emph{antichain} if none of its
sets is a proper subset of another set of~$\f$. For two vectors
$a\in\RR^n$ and $b\in\RR^n$ we write $a\leq b$ if $a_i\leq b_i$ holds
for all positions $i=1,\ldots,n$. A set $A$ of vectors is an
\emph{antichain} if $a\leq a'$ holds for no two distinct vectors
$a\neq a'\in A$. The \emph{characteristic vector} of a set
$S\subseteq[n]$ is the vector $a\in\{0,1\}^n$ with $a_i=1$ if and only
if $i\in S$.

\section{Preliminaries}
\label{sec:preliminaries}
Every finite set $A\subset\NN^n$ of \emph{feasible solutions} defines
a \emph{discrete optimization problem} of the form $f(x)=\min_{a\in
  A}\skal{a,x}$ or of the form $f(x)=\max_{a\in A}\skal{a,x}$, where
here and in what follows, $\skal{a,x}=a_1x_1+\cdots+a_nx_n$ stands for
the scalar product of vectors $a=(a_1,\ldots,a_n)$ and
$x=(x_1,\ldots,x_n)$.

We will refer to such problems as \emph{problems defined by} $A$, or
as \emph{problems on}~$A$. Such a problem is a \emph{$0$-$1$
  optimization problem} if the set $A\subseteq\{0,1\}^n$ of feasible
solutions consists of only $0$-$1$ vectors. These latter problems are
exactly what we called ``combinatorial optimization'' problems on
\emph{families} $\f\subseteq 2^{[n]}$ of feasible solutions, where~$\f$ consists of all sets $S_a=\{i\colon a_i=1\}$ for vectors $a\in
A$.

To avoid trivialities, we will throughout assume that the all-$0$
vector (or the empty set) is \emph{not} a feasible solution, that is,
we will always assume that $\vnul\not\in A$ and $\emptyset\not\in\f$.

\subsection{Circuits over semirings}
Recall that a (commutative) \emph{semiring} is a set $R$ closed under
two associative and commutative binary operations ``addition''
$(\suma)$ and ``multiplication'' $(\daug)$, where multiplication
distributes over addition: $x\daug(y\suma z)=(x\daug y)\suma (x\daug
z)$.  That is, in a semiring, we can ``add'' and ``multiply''
elements, but neither ``subtraction'' nor ``division'' are necessarily
possible. Besides of being commutative, we will assume that the
semiring contains a multiplicative identity element $\vienas$ with
$\vienas\daug x=x\daug\vienas=x$.

A \emph{circuit} over a semiring $R$ is a directed acyclic graph;
parallel edges joining the same pair of nodes are allowed.  Each
indegree-zero node (an \emph{input} node) holds either one of the
variables $x_1,\ldots,x_n$ or a semiring element.  Every other node, a
\emph{gate}, has indegree two and performs one of the semiring
operations. One of the gates is designated as the output gate. The
\emph{size} of a circuit is the total number of gates in it.  A
circuit is \emph{constant-free} if it has no semiring elements as
inputs.

Since in any semiring $(R,\suma,\daug)$, multiplication distributes
over addition, each circuit $\Phi$ over $R$ computes (at the output
gate) some polynomial
\begin{equation}\label{eq:pol}
  \Phi(x_1,\ldots,x_n)=\sum_{b\in B}\const{b} X^b \ \ \mbox{ with }\ \ X^b=\prod_{i=1}^n x_i^{b_i}\,
\end{equation}
over $R$ in a natural way, where $B\subset \NN^n$ is some set of
\emph{exponent vectors}, and $x_i^k$ stands for $x_i\daug
x_i\daug\cdots\daug x_i$ $k$-times.  Since we only consider semirings
with multiplicative identity, coefficients $\const{b}\in R$ are
semiring elements. To see why this assumption is necessary, consider
the semiring $(R,+,\times)$, where $R$ is the set of all positive even
integers. Then the coefficient $3$ of the monomial $x$ in the
polynomial $x+x+x$ is not a semiring element.

In this paper, we will mainly consider circuits over three commutative
and idempotent semirings $(R,\suma,\daug)$. In the \emph{boolean}
$(\lor,\land)$ semiring, we have $R=\{0,1\}$, $x\suma y:=x\lor y$ and
$x\daug y:=x\land y$. In the \emph{tropical} $(\min,+)$ semiring, we
have $R=\RR_+$, $x\suma y:=\min(x,y)$ and $x\daug y:=x+y$.  Similarly,
in the \emph{tropical} $(\max,+)$ semiring, we have $R=\RR_+$, $x\suma
y:=\max(x,y)$ and $x\daug y:=x+y$.  The multiplicative identity element
in the boolean semiring is $\vienas=1$, and is $\vienas=0$ in both
tropical semirings.  Over the boolean semiring, the polynomial
\cref{eq:pol} computes the monotone boolean function
\[
\Phi(x)=\bigvee_{b\in B} \bigwedge_{i:b_i\neq 0} x_i\,.
\]
Over the tropical semirings, every monomial $X^b=\prod_{i=1}^n
x_i^{b_i}$ turns into the scalar product
$X^b=\sum_{i=1}^nb_ix_i=\skal{b,x}$ of vectors $b$ and $x$. Hence, the
polynomial \cref{eq:pol} solves one of the two optimization problems
with linear objective functions:
\begin{equation}\label{eq:trop-pol}
  \Phi(x)=\min_{b\in B}\
  \skal{b,x}+\const{b}\ \mbox{ or }\ \Phi(x)=\max_{b\in B}\
  \skal{b,x}+\const{b}\,.
\end{equation}
Note that if a tropical circuit $\Phi$ is constant-free, then
$\const{b}=0$ holds for all $b\in B$.

\subsection{Sets of vectors produced by circuits}
\label{sec:produced}
A simple, but important in our later analysis, observation is that
every circuit of $n$ variables over a semiring $(R,\suma,\daug)$ not
only \emph{computes} some polynomial over $R$, but also
\emph{produces} (purely syntactically) a finite set of vectors in
$\NN^n$ in a natural way.

At each input node holding a semiring element, the same set
$\{\vec{0}\}$ is produced.  At an input node holding a variable $x_i$,
the set $\{\vec{e}_i\}$ is produced. At an ``addition'' $(\suma)$
gate, the union of sets produced at its inputs is produced. Finally,
at a ``multiplication'' $(\daug)$ gate, the Minkowski sum of sets
produced at its inputs is produced. The set produced by the entire
circuit is the set produced at its output gate.

It is clear that the same circuit $\Phi$ with only ``addition''
$(\suma$) and ``multiplication'' $(\daug)$ gates may \emph{compute}
different functions over different semirings. It is, however,
important to note that the set $B\subset\NN^n$ of vectors
\emph{produced} by $\Phi$ is always the same---it only depends on the
circuit itself, not on the underlying semiring.

On the other hand, up to coefficients, the polynomial function
\emph{computed} by the circuit $\Phi$ is determined by the set of
\emph{produced} vectors.

\begin{prop}\label{prop:function}
  If $B\subset\NN^n$ is the set of vectors produced by a circuit
  $\Phi$ over a semiring $R$, then $\Phi$ computes some polynomial
  over $R$ whose set of exponent vectors coincides with $B$.
\end{prop}

\begin{proof}
  Simple induction on the size of a circuit $\Phi$. Let
  $B\subset\NN^n$ be the set of vectors produced, and $f\colon R^n\to
  R$ the polynomial function computed by~$\Phi$.

  If the circuit $\Phi$ consists of a single input node holding a
  semiring element $\const{}\in R$, then $f(x)=\const{}$ is a constant
  polynomial with a single exponent vector $\vec{0}$. If $\Phi$
  consists of a single input node holding a variable $x_i$, then
  $f(x)=x_i$ is a degree-$1$ polynomial with the single exponent
  vector~$\vec{e}_i$.

  Now, the set of exponent vectors of a sum of two polynomials is just
  the union of the sets of exponent vectors of these polynomials.
  Finally, when multiplying two polynomials, we multiply each monomial
  of the first polynomial with all monomials of the second
  polynomial. The exponent vector of a product of two monomials is the
  sum of exponent vectors of these monomials.
\end{proof}

\bild

\begin{rem}
  In general, \cref{prop:function} has no converse, even for
  constant-free circuits: if a circuit $\Phi$ computes some polynomial
  $f$, then $\Phi$ does \emph{not} need to produce the set of exponent
  vectors of~$f$; a simple example for tropical circuits is given in
  \cref{fig:produced}.  Monotone arithmetic $(+,\times)$ circuits,
  that is, circuits over the arithmetic semiring $(\RR_+,+,\times)$,
  are here an exception: for them, also the converse of
  \cref{prop:function} holds. Namely, if such a circuit
  \emph{computes} a polynomial $f$, then the set of vectors
  \emph{produced} by the circuit is exactly the set of exponent
  vectors of this polynomial~$f$. This holds because, if two
  arithmetic polynomials coincide on sufficiently many (with respect
  to the number of variables and the degrees of these polynomials)
  inputs, then these polynomials must \emph{syntactically} coincide
  (even up to coefficients).
\end{rem}

\subsection{Eliminating constant inputs}
\label{sec:cf}
Recall that an optimization problem on a set $A\subset\NN^n$ of
feasible solutions is of the form $f(x)=\min_{a\in A}\skal{a,x}$ or of
the form $f(x)=\max_{a\in A}\skal{a,x}$. To avoid trivialities, we
always assume that $A\neq\emptyset$ and $\vnul\not\in A$.

These problems are ``constant-free'' in that they are
\emph{completely} specified by their sets $A$ of feasible solutions:
there are no additional constant terms.  In contrast, since tropical
circuits can have constant inputs, the optimization problems actually
solved by such circuits (exactly) may be not constant-free: they may
have additional constant terms; see \cref{eq:trop-pol}.

However, as the following lemma shows, when dealing with tropical
circuits approximating (constant-free) optimization problems, we can
safely restrict ourselves to constant-free circuits.  Recall that
constant-free circuits only use the variables $x_1,\ldots,x_n$ as
inputs.

\begin{lem}[Eliminating constant inputs]\label{lem:constant-free}
  If an optimization problem on a set $A\subset\NN^n$ can be
  $r$-approx\-i\-ma\-ted by a tropical circuit of size $t$, then this
  problem can also be $r$-approximated by a constant-free tropical
  circuit of size~$t$.
\end{lem}

\begin{proof}
  Let $\Phi$ be a tropical $(\max,+)$ or $(\min,+)$ circuit, and
  $B\subset\NN^n$ the set of vectors produced by $\Phi$. By
  \cref{prop:function}, the circuit computes the maximum or the
  minimum, over all vectors $b\in B$, of linear functions
  $\skal{b,x}+\const{b}$, where $\const{b}\in\RR_+$ are some
  constants.

  We obtain the \emph{constant-free version} $\Phi^*$ of $\Phi$ as
  follows. First, replace every constant input by~$0$. Then eliminate
  zeros by repeatedly replacing gates $u+0$ and $\max(u,0)$ by the
  gate $u$, and a gate $\min(u,0)$ by an input node holding $0$.
  Since $\Phi(x)\neq 0$ must hold for at least one $x\in\RR_+^n$, the
  constant $0$ input also disappears at the end of this replacement.
  Since constant inputs can only affect the additive constant terms
  $\const{b}$, the constant-free version $\Phi^*$ computes the maximum
  or the minimum of linear functions $\skal{b,x}$ \emph{without} any
  constant terms. Our goal is to show that $\Phi^*$ still
  $r$-approximates our optimization problem $f$ on the set~$A$.

  \case{1} $\Phi$ is a $(\max,+)$ circuit; hence, $f(x)=\max_{a\in A}
  \skal{a,x}$. In this case, we have that $\Phi^*(x)=\max_{b\in
    B}\skal{b,x}$, and $\Phi(x)=\max_{b\in B}\skal{b,x}+\const{b}$ for
  some nonnegative constants $\const{b}\in\RR_+$. Since $\Phi$
  approximates $f$, $\Phi(x)\leq f(x)$ must hold for all input
  weightings $x\in\RR_+^n$. Taking $x=\vnul$, we obtain
  $\Phi(\vnul)\leq f(\vnul)=0$ and, hence, $\const{b}=0$ for all $b\in
  B$. Thus, in the case of maximization, the constant-free version of
  the circuit solves just the same problem as the original circuit,
  and we are done.

  \case{2} $\Phi$ is a $(\min,+)$ circuit; hence, $f(x)=\min_{a\in A}
  \skal{a,x}$. Since $\Phi$ $r$-approximates $f$, we know that the
  inequalities $f(x)\leq \Phi(x)\leq r\cdot f(x)$ must hold for all
  $x\in\RR_+^n$.  We have to show that $\Phi^*$ also satisfies these
  inequalities. We know that $\Phi^*(x)=\min_{b\in B}\skal{b,x}$, and
  $\Phi(x)=\min_{b\in B}\skal{b,x}+\const{b}$ for some nonnegative
  constants $\const{b}\in\RR_+$.

  Since the constants $\const{b}$ are nonnegative, we clearly have
  $\Phi^*(x)\leq \Phi(x)$ and, hence, also $\Phi^*(x)\leq r\cdot f(x)$
  for all $x\in\RR_+^n$.  So, it remains to show that $\Phi^*(x)\geq
  f(x)$ holds for all $x\in\RR_+^n$, as well. We know that
  $\Phi(x)\geq f(x)$ holds for all $x\in\RR_+^n$.

  Assume contrariwise that $\Phi^*(\xx)< f(\xx)$ holds for some input
  weighting $\xx\in\RR_+^n$. Then the difference
  $\dff=f(\xx)-\Phi^*(\xx)$ is positive. We also know that
  $\const{}:=\max_{b\in B}\const{b}$ is positive, for otherwise, there
  would be nothing to prove. So, take the constant
  $\scalar:=2\const{}/\dff>0$, and consider the input weighting
  $\zz:=\scalar\cdot \xx$. Since $\Phi^*(\xx)=f(\xx)-\dff$, and since
  $\Phi(x)\leq \Phi^*(x)+\const{}$ holds for all weightings
  $x\in\RR_+^n$, the desired contradiction follows:
  \begin{align*}
    \Phi(\zz)&=\Phi(\scalar\cdot \xx)\leq \Phi^*(\scalar\cdot
    \xx)+\const{}=\scalar\cdot\Phi^*(\xx)+\const{}
    =\scalar\cdot [f(\xx)-\dff]+\const{}\\
    &=\scalar\cdot f(\xx)-\scalar\cdot\dff+\const{} = f(\scalar\cdot
    \xx)-\const{} =f(\zz)-\const{} < f(\zz)\,.
  \end{align*}
\end{proof}

\section{Approximation limitations of (min,+) circuits}
\label{sec:minplus}

In this section, we first prove a general ``boolean bound'' for
approximating $(\min,+)$ circuits: if the boolean (decision) version
of a minimization problem requires monotone boolean $(\lor,\land)$
circuits of size at least $t$, then no $(\min,+)$ circuit of size $<t$
can approximate the problem within any finite factor
(\cref{thm:bool}). Together with known lower bounds on the monotone
boolean circuit complexity, this gives us explicit minimization
problems which are hard to approximate by $(\min,+)$ circuits and,
hence, by pure DP algorithms; three selected examples are given in
\cref{sec:minplus-explicit}. Then, in \cref{sec:min-vs-greedy}, we
combine the boolean bound (\cref{thm:bool}) with counting arguments to
show that greedy algorithms can ``hardly'' beat pure DP algorithms:
there exist many minimization problems solvable by the greedy
algorithm exactly, while polynomial-size $(\min,+)$ circuits cannot
approximate any of them within any finite factor.

\subsection{The boolean bound for approximating $(\min,+)$ circuits}
\label{sec:bool-bound}
Recall that the \emph{support} of a vector $a\in\NN^n$ is the set
$\supp{a}=\{i\colon a_i\neq 0\}$ of its nonzero positions. Every
finite set $A\subset\NN^n$ of vectors \emph{defines} the monotone
boolean function
\[
f_A(x)= \bigvee_{a\in A}~\bigwedge_{i\in \supp{a}}x_i\,.
\]
Note that, for every input $x\in\{0,1\}^n$, we have
\begin{equation}\label{eq:bool}
  \mbox{$\bool{A}(x)=1$ if and only if
    $\supp{x}\supseteq \supp{a}$ for some $a\in A$.}
\end{equation}
For example, if $A$ is the set of characteristic $0$-$1$ vectors of
perfect matchings in $K_{m,m}$, then $f$ accepts a subgraph $G$ of
$K_{m,m}$ if and only if $G$ contains a perfect matching.

Two sets $A,B\subseteq\NN^n$ are \emph{similar} if the support of
every vector $b\in B$ contains the support of at least one vector
$a\in A$, and vice versa.  That is, $A$ and $B$ are similar if and
only if
\begin{equation}\label{eq:similar}
  \mbox{$\forall b\in B\ \exists a\in A\colon\, \supp{b}\supseteq\supp{a}$ and $\forall a\in A\ \exists b\in B\colon\, \supp{a}\supseteq\supp{b}$.}
\end{equation}
Observation~\cref{eq:bool} immediately yields the following.

\begin{prop}\label{prop:bool-str}
  Two sets of vectors define the same boolean function if and only if
  these sets are similar.
\end{prop}

The main connection between approximating $(\min,+)$ circuits and
monotone boolean circuits is given by the following lemma.  The
\emph{boolean version} of a constant-free tropical $(\min,+)$ circuit
is the monotone boolean $(\lor,\land)$ circuit obtained by replacing
each $\min$-gate by an $\lor$-gate, and each $+$-gate by an
$\land$-gate.

\begin{lem}\label{lem:minA}
  If a constant-free $(\min,+)$ circuit $\Phi$ approximates the
  minimization problem on a set $A\subset \NN^n$ within a finite
  factor~$r=r(n)\geq 1$, then the boolean version of $\Phi$ computes
  the boolean function defined by~$A$.
\end{lem}

\begin{proof}
  Let $B\subset \NN^n$ be the set of vectors produced by $\Phi$. Since
  the circuit $\Phi$ is constant-free, it solves the minimization
  problem $\Phi(x)=\min_{b\in B}\skal{b,x}$ defined by this set
  $B$. The minimization problem on $A$ is $f(x)=\min_{a\in
    A}\skal{a,x}$.  We know that $f(x)\leq \Phi(x)\leq r\cdot f(x)$
  must hold for all input weightings $x\in\RR_+^n$.  The boolean
  version $\bphi$ of $\Phi$ also produces the same set $B$. By
  \cref{prop:bool-str}, it remains to show that the set $B$ is similar
  to $A$; see~\cref{eq:similar}.

  For the sake of contradiction, suppose first that there is a vector
  $b\in B$ such that $\supp{a}\setminus \supp{b}\neq\emptyset$ holds
  for all vectors $a\in A$.  Consider the assignment $x\in\{0,1\}^n$
  of weights such that $x_i=0$ for $i\in\supp{b}$, and $x_i=1$ for
  $i\not\in\supp{b}$.  On this weighting, we have $\Phi(x)\leq
  \skal{b,x}=0$. But since every vector $a\in A$ has a position
  $i\not\in\supp{b}$ with $a_i\neq 0$, $\skal{a,x}\geq 1$ holds for
  all $a\in A$ and, hence, also $f(x)\geq 1$, contradicting the
  inequality $f(x)\leq \Phi(x)$.

  Now suppose that there is a vector $a\in A$ such that
  $\supp{b}\setminus \supp{a}\neq\emptyset$ holds for all vectors
  $b\in B$. Let $M=\max\{\skal{a,a}\colon a\in A\}$, and consider the
  weighting $x\in\{1,rM+1\}^n$ such that $x_i=1$ for all $i\in
  \supp{a}$ and $x_i=rM+1$ for all $i\not\in\supp{a}$ (note that
  $rM+1$ is a finite number, because both the approximation factor $r$
  and the set $A$ are finite.  Then $f(x)\leq
  \skal{a,x}=\skal{a,a}\leq M$. But since every vector $b\in B$ has a
  position $i\not\in\supp{a}$ such that $b_i\geq 1$, we have
  $\Phi(x)\geq rM+1 > r\cdot f(x)$, contradicting the inequality
  $\Phi(x)\leq r\cdot f(x)$.
\end{proof}

For a set $A\subseteq\NN^n$ of vectors, let $\Bool{A}$ denote the
minimum size of a monotone boolean $(\lor,\land)$ circuit computing
the boolean function $\bool{A}$ defined by~$A$. Let also $\Min{A}{r}$
denote the minimum size of a tropical $(\min,+)$ circuit approximating
the minimization problem on $A$ within the factor~$r$.

\begin{thm}[Boolean bound]\label{thm:bool}
  For every finite set $A\subset\NN^n$ and every finite factor
  $r=r(n)\geq 1$, we have $\Min{A}{r}\geq \Bool{A}$.
\end{thm}

\begin{proof}
  Take a $(\min,+)$ circuit $\Phi$ of size $t=\Min{A}{r}$
  approximating the minimization problem on $A$ within the factor $r$,
  and let $B\subset\NN^n$ be the set of vectors produced by $\Phi$. By
  \cref{lem:constant-free}, we can assume that the circuit $\Phi$ is
  constant-free.  Hence, by \cref{lem:minA}, the boolean version
  $\bphi$ of $\Phi$ (which has the same size) computes the boolean
  function defined by the set~$A$, as desired.
\end{proof}

\begin{rem}
  Note that \cref{thm:bool} does \emph{not} exclude that, using
  \emph{more} than $\Bool{A}$ gates, $(\min,+)$ circuits could achieve
  finite (and even small) approximation factors. The boolean bound
  $\Bool{A}$ is just an ``absolute'' lower bound below which no
  approximation is possible at all.
\end{rem}

\begin{rem}\label{rem:semdeg}
  The proof of \cref{thm:bool} is so direct and elementary, because it
  totally ignores the given approximation factor $r$: it only must be
  finite and, hence, can be used in input weightings to fool too small
  $(\min,+)$ circuits. Using more involved arguments (based on Farkas'
  lemma), we will show in \cref{sec:bool-tight} (\cref{thm:minexact})
  that, under an appropriate definition of the ``semantic degree'' of
  monotone boolean circuits, \cref{thm:bool} has also a
  \emph{converse}: a minimization problem can be approximated within a
  factor $r$ by a tropical $(\min,+)$ circuit of size~$t$ if and only
  if the boolean version of this problem can be computed by a monotone
  boolean $(\lor,\land)$ circuit of size $t$ and semantic degree at
  most~$r$.  Thus, the \emph{approximation} power of tropical
  $(\min,+)$ circuits is \emph{captured}, not only lower bounded, by
  the \emph{computational} power of monotone boolean circuits.
\end{rem}

\subsection{Explicit lower bounds}
\label{sec:minplus-explicit}

Together with lower bounds on the monotone boolean circuit complexity,
the boolean bound (\cref{thm:bool}) immediately yields the same lower
bounds on the size of approximating $(\min,+)$ circuits.  Let us
mention some examples.

In the \emph{lightest triangle} problem, we are given an assignment of
nonnegative weights to the edges of $K_n$, and the goal is to compute
the minimum weight of a triangle.

\begin{cor}\label{cor:triangle}
  The lightest triangle problem in $K_n$ can be solved by a $(\min,+)$
  circuit using only $n^3$ gates, but no $(\min,+)$ circuit with
  $n^{3-\Omega(1)}$ gates can approximate this problem within any
  finite factor.
\end{cor}

\begin{proof}
  Since we only have $\binom{n}{3}$ triangles, a trivial $(\min,+)$
  circuit of size at most $n^3$ (taking the minimum over all
  triangles) solves this problem exactly. On the other hand, it is
  known (\cite[Lemma~3.14]{AB87}) that the decision version of this
  problem requires monotone boolean circuits with
  $\Omega(n^3/\log^4n)$ gates. \Cref{thm:bool} gives the same lower
  bound for approximating $(\min,+)$ circuits.
\end{proof}

Recall that the $n$-\emph{assignment problem} is: given an assignment
of nonnegative real weights to the edges of the complete bipartite
$n\times n$ graph, compute the minimum weight of a perfect matching in
this graph.  The corresponding family of feasible solutions is here
the family of all perfect matchings, viewed as sets of their edges.

\begin{cor}\label{cor:match}
  Every $(\min,+)$ circuit approximating the $n$-assignment problem
  within any finite factor must have at least $n^{\Omega(\log n)}$
  gates.
\end{cor}

\begin{proof}
  The boolean function defined by the family of feasible solutions of
  the assignment problem is the boolean permanent function which, as
  proved by Raz\-bo\-rov~\cite{razb-perm}, requires monotone boolean
  circuits of size $n^{\Omega(\log n)}$.
\end{proof}

Let $n$ be a prime power, and $1\leq d\leq n$ an integer.  The
\emph{polynomial $(n,d)$-design} is the family of all $|\f|=n^d$
$n$-element subsets $\{(a,p(a))\colon a\in \gf{n}\}$ of the grid
$\gf{n}\times\gf{n}$, where $p=p(x)$ ranges over all $n^d$ univariate
polynomials of degree at most $d-1$ over $\gf{n}$.

\begin{cor}\label{cor:andreev}
  If $d\leq (n/4\ln n)^{1/2}$, then every $(\min,+)$ circuit
  approximating the minimization problem on the polynomial
  $(n,d)$-design within any finite factor must have at least
  $n^{\Omega(d)}$ gates.
\end{cor}

\begin{proof}
  By (numerically) improving the earlier result of
  Andreev~\cite{andreev85}, Alon and Boppana~\cite{AB87} have shown
  that, at least for such values of $d$, any monotone boolean circuit
  computing the boolean function defined by the corresponding family
  of feasible solutions requires $n^{\Omega(d)}$ gates.
\end{proof}

\subsection{Greedy can beat approximating (min,+) circuits}
\label{sec:min-vs-greedy}
Our goal now is to show that there \emph{exist} many (combinatorial)
minimization problems which are solvable by the greedy algorithm
exactly (within the factor $r=1$), but no $(\min,+)$ circuit with a
polynomial in the number $n$ of ground elements number of gates can
approximate any of these problems within \emph{any} finite factor
$r=r(n)$.

We identify matroids with their families of bases. Under this proviso,
a family $\f$ is a matroid if and only if $\f$ is uniform (all sets
have the same cardinality) and the \emph{basis exchange axiom} holds:
if $A\neq B\in \f$, then for every $a\in A\setminus B$ there is a
$b\in B\setminus A$ such that the set $(A\setminus\{a\})\cup\{b\}$
belongs to $\f$.

It is well known (see, for example, \cite[Theorem~1.8.4]{oxley}) that
an optimization problem on an antichain $\f$ can be solved by the
greedy algorithm exactly if and only if $\f$ is a matroid. This fact
is usually called the Rado--Edmonds theorem~\cite{Rado,Edmonds}.  In
contrast, we will now show that most matroids require $(\min,+)$
circuits of exponential size to be even only approximated within any
finite factor. We will do this by counting, so we need a lower bound
on the number of matroids.

The following simple construction of matroids was implicit in several
papers, starting from those of Piff and Welsh~\cite{piff}, and
Knuth~\cite{knuth}, and was made explicit by Bansal, Pendavingh and
Van der Pol~\cite[Lemma~8]{bansal}. Let $\binom{[n]}{m}$ denote the
family of all $m$-element subsets of $[n]=\{1,\ldots,n\}$. The Hamming
distance between two sets $A$ and $B$ is $\dist{A,B}=|A\setminus
B|+|B\setminus A|$. A~family $\h$ is \emph{separated} if
$\dist{A,B}>2$ holds for all $A\neq B\in\h$.

\begin{prop}\label{prop:sep}
  If $\h\subseteq \binom{[n]}{m}$ is separated, then
  $\f=\binom{[n]}{m}\setminus \h$ is a matroid.
\end{prop}

\begin{proof}
  Suppose contrariwise that $\f$ is not a matroid. Since the family
  $\f$ is uniform, there must be two sets $A\neq B\in\f$ violating the
  basis exchange axiom: there is an $a\in A\setminus B$ such that
  $(A\setminus\{a\})\cup\{b\}\not\in\f$ for all $b\in B$.  Observe
  that $B\setminus A$ must have at least two elements: held
  $B\setminus A=\{b\}$ then, since both $A$ and $B$ have the same
  cardinality, the set $(A\setminus\{a\})\cup\{b\}$ would coincide
  with $B$ and, hence, would belong to $\f$. So, take $b\neq c\in
  B\setminus A$ and consider the sets $S=(A\setminus\{a\})\cup\{b\}$
  and $T=(A\setminus\{a\})\cup\{c\}$. Since the basis exchange axiom fails
  for $A$ and $B$, neither $S$ nor $T$ can belong to $\f$; hence,
  \emph{both} sets $S$ and $T$ belong to the family
  $\binom{[n]}{m}\setminus \f=\h$.  But $\dist{S,T}=|\{b,c\}|=2$, a
  contradiction with the family $\h$ being separated.
\end{proof}

\begin{prop}\label{prop:graham}
  There are $2^{\binom{n}{m}/n}$ matroids $\f\subseteq \binom{[n]}{m}$
  such that $\binom{[n]}{m}\setminus \f$ is separated.
\end{prop}

\begin{proof}
  Since subfamilies of separated families are also separated, it is
  enough, by \cref{prop:sep} to show that a separated family
  $\h\subseteq \binom{[n]}{m}$ of size $|\h|\geq \binom{n}{m}/n$
  exists.

  The following amazingly simple argument was suggested by Graham and
  Sloane~\cite{graham}. For $l\in\{0,1,\ldots,n-1\}$, let $\h_l$ be
  the family of all sets $S\in \binom{[n]}{m}$ such that $\sum_{i\in
    S}i=l\mod{n}$.  We claim that each such family $\h_l$ is
  separated. Suppose contrariwise that $\dist{S,T}=2$ holds for some
  two sets $S\neq T$ of $\h_l$. Then $S=A\cup\{s\}$ and $T=A\cup\{t\}$
  for some $(m-1)$-element set $A$, and $s\neq t$ are distinct numbers
  in $[n]\setminus A$. But then for $a=\sum_{i\in A}i$, we have
  $a+s=l\mod{n}$ and $a+t=l\mod{n}$, which is impossible because both
  numbers $s$ and $t$ are at most $n$.  Thus, every family $\h_l$ is
  separated. Since there are only $n$ such families, and they exhaust
  the entire family $\binom{[n]}{m}$, there must be an $l$ for which
  $|\h_l|\geq \binom{n}{m}/n$ holds.
\end{proof}

\begin{thm}\label{thm:min-gap}
  There are at least $2^{2^n/n^3}$ matroids $\f\subseteq 2^{[n]}$ such
  that every $(\min,+)$ circuit approximating the minimization problem
  on any of them within any finite factor $r=r(n)\geq 1$ must have at
  least $2^n/n^3$ gates.
\end{thm}

\begin{proof}
  The number of monotone boolean $(\lor,\land)$ circuits of size $t$
  on $n$ input variables is at most $L(n,t)=2^t(t+n)^{2t}$. This is,
  actually, an upper bound on the number of constant-free circuits
  over any semiring $(R,\suma,\daug)$.  Indeed, each gate in such a
  circuit is assigned a semiring operation (two choices) and acts on
  two previous nodes. Each previous node can either be a previous gate
  (at most $t$ choices) or an input variable ($n$ choices). Thus, each
  single gate has at most $N=2(t+n)^2$ choices, and the number of
  choices for a circuit is at most $N^t$.

  When applied with $m=\lfloor n/2\rfloor$, \cref{prop:graham} gives
  us at least $M(n)=2^{\binom{n}{m}/n}\geq 2^{2^{n}/2n^{3/2}}$
  matroids $\f\subseteq \binom{[n]}{m}$. On the other hand, at most
  $L(n,t)$ families $\f\subseteq 2^{[n]}$ can have monotone boolean
  circuit complexity at most~$t$. For $t:=2^n/n^3$, we have $\log
  L(n,t)=2^n/n^3+(2^{n+1}/n^3)\log(n+2^n/n^3)=O(2^n/n^2)\ll \log
  M(n)=2^n/2n^{3/2}$.  Since every circuit computes only one function,
  at least $M(n)-L(n,t)\geq L(n,t)$ matroids require monotone boolean
  circuits of size at least~$t=2^n/n^3$.  \Cref{thm:bool} yields the
  same lower bound for approximating $(\min,+)$ circuits.
\end{proof}

\section{Approximation limitations of (max,+) circuits}
\label{sec:maxplus}

Given a family $\f\subseteq 2^{[n]}$ of feasible solutions, and an
approximation factor $r\geq 1$, we will denote by $\Max{\f}{r}$ the
minimum number of gates in a $(\max,+)$ circuit approximating the
maximization problem $f(x)=\max_{S\in \f}\sum_{i\in S}x_i$ on $\f$
within the factor~$r$.

In \cref{sec:minplus}, we have shown that there are (even explicit)
families $\f\subseteq 2^{[n]}$, the \emph{minimization} problems on
which cannot be approximated by small (polynomial in $n$) size
$(\min,+)$ circuits within any finite factor $r=r(n)$.  On the other
hand, in the case of \emph{maximization} problems, the approximation
factor is always finite.  Namely, we always have $\Max{\f}{n}\leq
n-1$: since the weights are nonnegative, we can just use the trivial
$(\max,+)$ circuit $\max\{x_1,\ldots,x_n\}$.

\subsection{Counting fails for approximating $(\max,+)$ circuits}
\label{sec:count-fails}

There is an even more substantial difference between approximating
$(\min,+)$ and $(\max,+)$ circuits than just the ``bounded versus
unbounded approximation factors'' phenomenon: unlike for $(\min,+)$
circuits, even \emph{counting} arguments are unlikely to yield large
lower bounds on the size of approximating $(\max,+)$ circuits, even
for very small approximation factors $r=1+o(1)$.

Say that a family $\f\subseteq 2^{[n]}$ is $k$-\emph{dense} if every
$k$-element subset of $[n]$ is contained in at least one set of~$\f$.
The \emph{top $k$-of-$n$ selection} problem
$\sel{n}{k}(x_1,\ldots,x_n)$ outputs the sum of the $k$ largest input
numbers.

\begin{prop}\label{prop:dense}
  The top $k$-of-$n$ selection problem $\sel{n}{k}$ can be solved by a
  $(\max,+)$ circuit of size $2kn$, and this circuit approximates the
  maximization problem on every $k$-dense family
  $\f\subseteq\binom{[n]}{m}$ within the factor $r=m/k$.
\end{prop}

\begin{proof}
  The family of feasible solutions of $\sel{n}{k}$ consists of all
  $k$-element subsets of $[n]=\{1,\ldots,n\}$. In particular,
  $\sel{n}{1}(x)=\max\{x_1,\ldots,x_n\}$ and $\sel{n}{n}(x)=
  x_1+\cdots+x_n$. The Pascal identity
  $\binom{n+1}{k}=\binom{n}{k}+\binom{n}{k-1}$ for binomial
  coefficients gives us the recursion
  \[
  \sel{n+1}{k}(x_1,\ldots,x_{n+1})=\max\{\sel{n}{k}(x_1,\ldots,x_{n}),
  \sel{n}{k-1}(x_1,\ldots,x_{n})+x_{n+1}\}\,.
  \]
  So, $\sel{n}{k}$ can be solved by a $(\max,+)$ circuit with only
  $2kn$ $(\max,+)$ gates.

  Now let $\f\subseteq\binom{[n]}{m}$ be a $k$-dense family.  The
  maximization problem on $\f$ is $f(x)=\max_{S\in\f}\sum_{i\in
    S}x_i$. Since the weights are nonnegative, the $k$-denseness of
  $\f$ ensures that $f(x)\geq \sel{n}{k}(x)$.  On the other hand,
  since no solution has more than $m$ elements, the optimal weight
  $f(x)$ of a feasible solution cannot exceed $m/k$ times the sum of
  weights of $k$ heaviest elements in this solution. Hence, $f(x)\leq
  (m/k)\cdot \sel{n}{k}(x)$, as desired.
\end{proof}

\begin{prop}\label{prop:counting1}
  There exist doubly-exponentially many in $n$ families $\f\subseteq
  2^{[n]}$ such that $\Max{\f}{1}=2^{\Omega(n)}$ but
  $\Max{\f}{1+o(1)}\leq n^2$.
\end{prop}
The families $\f$ are here \emph{matroids}, and the upper bound is
achieved by \emph{one single} $(\max,+)$ circuit.

\begin{proof}
  Let $n$ be a sufficiently large even integer, and $m=n/2$.
  \Cref{prop:graham} gives us at least $M:=2^{\binom{n}{m}/n}$
  families $\f\subseteq\binom{[n]}{m}$ (which are matroids) with the
  property that the Hamming distance between any two distinct sets
  $A\neq B\in \binom{[n]}{m}\setminus\f$ is $>2$. We claim that each
  such family $\f$ is $k$-dense for $k:=m-1$. To see this, take any
  set $T\in \binom{[n]}{m-1}$, any two distinct elements $a\neq b$
  outside $T$, and consider the $m$-element sets $A=T\cup\{a\}$ and
  $B=T\cup\{b\}$.  Since the Hamming distance between $A$ and $B$ is
  $2$, they cannot \emph{both} lie outside the family $\f$. So, at
  least one of them must belong to $\f$, as desired.

  We thus have at least $M$ families $\f\subseteq\binom{[n]}{m}$ which
  are $k$-dense for $k=m-1$.  By \cref{prop:dense}, one $(\max,+)$
  circuit of size at most $2kn\leq n^2$ for the top $k$-of-$n$ problem
  $\sel{n}{k}$ approximates the maximization problem on \emph{each} of
  these $M$ families within the factor $r=m/(m-1)=1+1/(m-2)=1+o(1)$.
  On the other hand, by \cref{lem:constant-free}, we can consider only
  constant-free $(\max,+)$ circuits, and the same counting argument as
  in the proof of \cref{thm:min-gap} yields the lower bound
  $\Max{\f}{1}=2^{\Omega(n)}$ for doubly-exponentially many of these
  families~$\f$.
\end{proof}

\begin{rem}[Boolean bound fails for approximating $(\max,+)$]\label{rem:bool-max}
  The standard counting (as in the proof of \cref{thm:min-gap}) shows
  that the boolean function defined by some of the families $\f$ given
  by \cref{prop:counting1} (actually, by many of these families)
  requires monotone boolean circuits of size $2^{\Omega(n)}$, but (by
  \cref{prop:counting1}) $\Max{\f}{r}\leq n^2$ holds already for a
  factor $r=1+o(1)$.
\end{rem}

Actually, small $(\max,+)$ circuits can approximate even \emph{random}
maximization problems quite well.  For an even integer $n\geq 4$ and
$m=n/2$, let $\xf$ be a random family of $m$-element subsets of $[n]$
with each $m$-element subset being included in $\xf$ independently
with probability $1/2$.

\begin{prop}\label{prop:counting2}
  With probability $1-o(1)$, $\Max{\xf}{1+o(1)}\leq n^2$ holds.
\end{prop}

\begin{proof}
  Let $k:=m-2$. Since each $k$-element set is contained in
  $l=\binom{n-k}{2}=\Omega(n^2)$ sets of $\binom{[n]}{m}$, the
  probability that a fixed $k$-element set will be contained in
  \emph{none} of the sets of $\xf$ is $(1/2)^l=2^{-\Omega(n^2)}$.  So,
  by the union bound, the family $\xf$ is \emph{not} $k$-dense with
  probability at most $\binom{n}{m}\cdot
  2^{-\Omega(n^2)}=2^{-\Omega(n^2)}$.  That is, the family $\xf$ is
  $k$-dense with probability at least $1-2^{-\Omega(n^2)}$.  By
  \cref{prop:dense}, with this probability, the $(\max,+)$ circuit for the top $k$-of-$n$ selection problem (see \cref{prop:dense})  approximates the maximization problem on a random family
  $\xf\subseteq \binom{[n]}{n/2}$ within the factor $r=m/k=1+o(1)$.
\end{proof}

\Cref{prop:counting1,prop:counting2} only show the mere
\emph{existence} of maximization problem that are hard to solve by
$(\max,+)$ circuits exactly (with factor $r=1$), but can be
approximated by small $(\max,+)$ circuits within a slightly larger
factor $r>1$.  Still, there are also \emph{explicit} maximization
problems exhibiting a similar gap.

A family $\f\subseteq 2^{[n]}$ is a \emph{Sidon family} if the set
$A\subseteq\{0,1\}^n$ of the characteristic $0$-$1$ vectors of sets in
$\f$ has the following property for all vectors $a,b,c,d\in A$: if
$a+b=c+d$, then $\{a,b\}=\{c,d\}$ (the addition is over the reals, not
over $\gf{2}$). That is, knowing the sum $a+b\in\{0,1,2\}^n$ of two
vectors $a,b\in A$, we know which vectors were added.

\begin{thm}[Explicit gaps]\label{thm:sidon}
  Let $m$ be an odd integer, and $n=4m$. Then there is an explicit
  Sidon family $\f\subseteq 2^{[n]}$ such that $\Max{\f}{1}\geq
  2^{n/4}$ but $\Max{\f}{2}\leq n$.
\end{thm}

The lower bound here follows from known lower bounds on the monotone
arithmetic circuit complexity of polynomials whose sets of exponent
vectors are Sidon sets, but the proof of the upper bound is somewhat
technical. So, since we are mainly interested in proving \emph{lower}
bounds, we postpone the entire proof of \cref{thm:sidon} to
\cref{app:sidon}.

The message of \cref{prop:counting1,prop:counting2} is: while
\emph{most} problems are hard to solve exactly, they are trivially
approximable by just \emph{one} small $(\max,+)$ circuit within a
small factor $r=1+o(1)$. Together with \cref{thm:sidon} (and
\cref{rem:bool-max}), this serves as a serious indication that the
task of proving lower bounds on the size of approximating $(\max,+)$
circuits is by far more difficult than for $(\min,+)$ circuits:
monotone \emph{boolean} circuits cannot help then, and even counting
arguments are unlikely to work against $(\max,+)$ circuits.

Still, by looking more carefully into the structure of vectors
\emph{produced} by approximating $(\max,+)$ circuits
(\cref{lem:maxA}), and using structural restrictions of such sets
given by a ``decomposition lemma'' (\cref{lem:decomp}), we will be
able to derive a general ``rectangle bound'' for approximating
$(\max,+)$ circuits (\cref{thm:rect}).

\subsection{Structure of approximating $(\max,+)$ circuits}
\label{sec:structure}
Since we are interested in the structure of sets of \emph{vectors}
produced by (approximating) circuits, it will be convenient to turn to
the language of \emph{vectors}.

\begin{lem}\label{lem:maxA}
  If $\Phi$ is a $(\max,+)$ circuit approximating the maximization
  problem on a set $A\subseteq \{0,1\}^n$ within a factor $r$, then
  the set $B\subset \NN^n$ of vectors produced by $\Phi$ has the
  following two properties:
  \begin{itemize}
  \item[\mbox{\rm (i)}] if $b\in B$, then $b\leq a$ for some $a\in A$;
  \item[\mbox{\rm (ii)}] if $a\in A$, then $\skal{a,b}\geq
    \frac{1}{r}\skal{a,a}$ for some $b\in B$.
  \end{itemize}
\end{lem}

\begin{proof}
  By \cref{lem:constant-free}, we can assume that the circuit $\Phi$
  is constant-free.  By \cref{prop:function}, the circuit $\Phi$
  solves the maximization problem $\Phi(x)=\max_{b\in B}\skal{b,x}$.
  The maximization problem on $A$ is of the form $f(x)=\max_{a\in A}\
  \skal{a,x}$.  Since the circuit $r$-approximates the maximization
  problem on $A$, we know that $\frac{1}{r}\cdot f(x)\leq \Phi(x)\leq
  f(x)$ must hold for all input weightings $x\in\RR_+^n$.

  Had some vector $b\in B$ a position $i$ with $b_i>1$, then on the
  input $x=\vec{e}_i$, we would have $\Phi(x)\geq \skal{b,x}=b_i>1$
  but $f(x)\leq 1$, since all vectors in $A$ are $0$-$1$ vectors. So,
  $B\subseteq\{0,1\}^n$, that is, the set $B$ also consists of only
  $0$-$1$ vectors.

  To show item (i), suppose contrariwise that there is a vector $b\in
  B$ such that $b\not\leq a$ holds for all vectors $a\in A$.  Since
  (as we have just shown) $b$ is a $0$-$1$ vector, this means that,
  for every vector $a\in A$, there is a position $i$ where $b_i=1$ but
  $a_i=0$.  Hence, on the weighting $x:=b$, we have $\Phi(x)\geq
  \skal{b,x}=\skal{b,b}$, but $\skal{a,x}=\skal{a,b}\leq \skal{b,b}-1$
  for all $a\in A$, a contradiction with $\Phi(x)\leq f(x)$.

  To show item (ii), assume contrariwise that there is some vector
  $a\in A$ such that $\skal{a,b}< m/r$ holds for all vectors $b\in B$,
  where $m=\skal{a,a}$. Then, on the input $x:=a$, we have $\Phi(x) <
  m/r$ but $f(x)\geq \skal{a,a}= m$, a contradiction with
  $\frac{1}{r}\cdot f(x)\leq \Phi(x)$.
\end{proof}

\begin{rem}\label{rem:boolean-weights}
  Note that \cref{lem:maxA} holds even when the circuit $\Phi$ is only
  required to $r$-approximate the given minimization problem on input
  weightings $x\in\{0,1\}^n$.  Indeed, to eliminate constant inputs
  from $(\max,+)$ circuits in \cref{lem:constant-free} we only used
  the input weighting $x=\vnul$, and the proof of \cref{lem:maxA}
  itself also uses only boolean $0$-$1$ weightings. This implies that
  the rectangle bound (\cref{thm:rect}), whose proof will use
  \cref{lem:maxA}, holds also when the $(\max,+)$ circuits must
  approximate a given problem only on boolean $0$-$1$ weightings.
\end{rem}

\subsection{Minkowski circuits}
\label{sec:minkowski}
As we already mentioned in \cref{sec:produced}, unlike the function
\emph{computed} by a circuit $\Phi$ over a semiring $(R,\suma,\daug)$,
the set $B\subset\NN^n$ of vectors \emph{produced} by $\Phi$ does not
depend on the underlying semiring---it only depends on the circuit
itself.  That is, $B$ depends only on what the underlying graph of
$\Phi$ is, and what of the two semiring operations are associated with
gates.  This independence of produced sets from actual semirings is
captured by the model of ``Minkowski circuits.'' These circuits allow
one to analyze the \emph{structure} of sets produced by circuits over
arbitrary semirings in a uniform and mathematically clean way.

A \emph{Minkowski circuit} $\Phi$ is a directed acyclic graph with
$n+1$ input (indegree zero) nodes holding single-element sets
$\{\vec{0}\},\{\vec{e}_1\},\ldots,\{\vec{e}_n\}$. Every other node, a
\emph{gate}, has indegree two, and performs either the set-theoretic
union ($\cup$) or the Minkowski sum ($+$) operation on its two inputs.

The sets $\pr{v}\subset \NN^n$ of vectors \emph{produced} at the gates
$v$ of $\Phi$ are obtained as follows. If $v$ is an input node, then
$\pr{v}$ is one of the single-element sets
$\{\vec{0}\},\{\vec{e}_1\},\ldots,\{\vec{e}_n\}$, depending on which
of these sets is held by the node~$v$.  Then $\pr{v}=\pr{u}\cup\pr{w}$
if $v=u\cup w$ is a union gate, and $\pr{v}=\pr{u}+\pr{w}$ if $v=u+w$
is a Minkowski sum gate. The set $B\subset \NN^n$ produced by the
entire circuit $\Phi$ is the set $\pr{v}$ produced at the output
gate~$v$.

The \emph{Minkowski version} of a circuit $\Phi$ over an arbitrary
semiring $(R,\suma,\daug)$ is obtained by replacing each input
constant $\const{}\in R$ by the singleton $\{\vnul\}$, each input
variable $x_i$ by the singleton $\{\vec{e}_i\}$, each ``addition''
$(\suma)$ gate by the union $(\cup)$ gate, and each ``multiplication''
$(\daug)$ gate by the Minkowski sum $(+)$ gate.

The model of Minkowski circuits is justified by the following trivial
observation, which follows directly from the definition of sets
produced by circuits over semirings: the set produced by a circuit
over \emph{any} semiring is the set produced by the Minkowski version
of this circuit.

\subsection{Decomposition lemma for Minkowski circuits}
\label{sec:covering}
We will prove lower bounds for approximating $(\max,+)$ circuits using
a general ``decomposition lemma'' for Minkowski circuits.  The
\emph{sumset} defined by two sets $X\subseteq\NN^n$ and $Y\subseteq
\NN^n$ of vectors is the Minkowski sum $X+Y=\{x+y\colon x\in X, y\in
Y\}$ of these two sets.

Sumsets naturally emerge in every Minkowski circuit $\Phi$.  At each
Minkowski sum gate following a gate $v$ (if there is any), the set
$\pr{v}$ of vectors produced at $v$ is ``enlarged'' by adding at least
one vector to \emph{all} vectors in $\pr{v}$.  So, when we arrive at
the output gate $w$, the entire translates $\pr{v}+y=\{x+y\colon
x\in\pr{v}\}$ of $\pr{v}$ by some vectors $y\in\NN^n$ must lie in the
set $\pr{w}=B$ produced by the entire circuit. This observation
motivates to associate with every gate $v$ its \emph{residue}
\[
\compl{v}=\{y\in\NN^n\colon \pr{v}+y\subseteq B\}
\]
which collects all vectors $y\in\NN^n$, the translates of $\pr{v}$ by
which lie in the set $B$ produced by the entire circuit.  For example,
if $v$ is the output gate, then $\pr{v}=B$ and
$\compl{v}=\{\vec{0}\}$. If $v$ is an input node, then either
$\pr{v}=\{\vec{0}\}$ and $\compl{v}=B$, or $\pr{v}=\{\vec{e}_i\}$ and
$\compl{v}=\{b-\vec{e}_i\colon b\in B, b_i\geq 1\}$.

Note that neither $\pr{v}$ nor $\compl{v}$ needs lie in $B$, but
$\pr{v}+\compl{v}\subseteq B$ already holds for every gate~$v$. Thus,
if the circuit $\Phi$ has $t$ gates, then we obtain a covering of the
set $B$ by $t$ sumsets of the form $\pr{v}+\compl{v}$

A~\emph{norm-measure} is any assignment $\Nor:\NN^n\to\RR_+$ of
nonnegative real numbers to vectors in $\NN^n$ such that every $0$-$1$
vector with at most one $1$ gets norm at most $1$, and the norm is
monotone and subadditive: $\nor{x}\leq \nor{x+y}\leq \nor{x}+\nor{y}$
holds for all vectors $x,y\in\NN^n$.

\begin{lem}[Decomposition lemma~\cite{juk-SIDMA}]\label{lem:decomp}
  If a set $B\subset\NN^n$ can be produced by a Minkowski $(\cup,+)$
  circuit of size $t$, then $B$ is a union of $t$ sumsets
  $X+Y\subseteq B$ with the following property.
  \begin{enumerate}
  \item[$(\ast)$] \label{itm:decom} For every norm-measure
    $\Nor:\NN^n\to\RR_+$, for every vector $b\in B$ of norm
    $\nor{b}>1$, and every $1/\nor{b}\leq \bbal < 1$ at least one of
    these sumsets $X+Y$ contains vectors $x\in X$ and $y\in Y$ such
    that $x+y=b$ and
    \[
    \tfrac{\bbal}{2}\cdot\nor{b}<\nor{x}\leq \bbal\cdot \nor{b}\,.
    \]
  \end{enumerate}
\end{lem}

The lemma was originally proved in \cite[Theorem~D]{juk-SIDMA}. Here
we give a simpler proof.

\begin{proof}
  Let $\Phi$ be a Minkowski $(\cup,+)$ circuit of size $t$ producing
  the set $B$.  Since we have only $t$ gates in the circuit, it is
  enough to show that the collection of sumsets $\pr{v}+\compl{v}$
  associated with the gates $v$ of $\Phi$ has the desired
  property~$(\ast)$.  So, fix some norm-measure $\mu:\NN^n\to\RR_+$,
  some vector $b\in B$ of norm $\mm:=\nor{b}>1$, and a real number
  $1/\mm\leq \bbal < 1$.

  By a \emph{decomposition} of the vector $b$ (or just a
  \emph{decomposition}, because the vector $b$ is fixed) at a gate $v$
  we will mean a pair $(x,y)\in\pr{v}\times\compl{v}$ of vectors (if
  there is one) such that $x+y=b$. The \emph{norm} of such a
  decomposition is the norm $\nor{x}$ of the first vector (that in the
  set $\pr{v}$).  Note that at the output gate, we have the unique
  decomposition $(x,y)=(b,\vec{0})$ of $b$ of
  norm~$\nor{x}=\nor{b}=\mm$.

\begin{clm}
  Let $v$ be a gate entered from gates $u$ and $w$. If there is a
  decomposition $(x,y)$ of vector $b$ at gate $v$, then there is a
  decomposition $(x',y')$ of $b$ at $u$ or $w$ such that
  $\tfrac{1}{2}\cdot\nor{x}\leq \nor{x'}\leq \nor{x}$.
\end{clm}

 \begin{proof}
   If $v=u\cup w$ is a union gate, then $\pr{v}=\pr{u}\cup \pr{w}$
   and, hence, $\compl{v}=\compl{u}\cap \compl{w}$.  So, the same pair
   $(x,y)$ is a decomposition at the gate $u$ (if $x\in\pr{u}$) or at
   the gate $w$ (if $x\in\pr{w}$), and the claim is trivial in this
   case.

   Assume now that $v=u+w$ is a Minkowski sum gate.  Then $x=x_u+x_w$
   for some vectors $x_u\in\pr{u}$ and $x_w\in\pr{w}$. Since vector
   $y$ belongs to the residue $\compl{v}$ of gate $v$, we know that
   $\pr{u}+\pr{w}+y\subseteq B$ holds. In particular, both inclusions
   $\pr{u}+(x_w+y)\subseteq B$ and $\pr{w}+(x_u+y)\subseteq B$ must
   hold. So, vector $x_w+y$ belongs to the residue $\compl{u}$ of gate
   $u$, and vector $x_u+y$ belongs to the residue $\compl{w}$ of gate
   $w$.  This implies that the pair $(x_u,x_w+y)$ is a decomposition
   of $b$ at the gate $u$, and the pair $(x_w,x_u+y)$ is a
   decomposition of $b$ at the gate $w$. Since $x=x_u+x_w$, the
   monotonicity of the norm implies that both $\nor{x_u}$ and
   $\nor{x_w}$ are at most $\nor{x}$, while the subadditivity of the
   norm implies that one of the norms $\nor{x_u}$ and $\nor{x_w}$ of
   these decompositions must be at least $\tfrac{1}{2}\cdot\nor{x_u+x_w}=\tfrac{1}{2}\cdot\nor{x}$,
   and we can take that input $u$ or $w$ at which the decomposition
   has larger norm.
 \end{proof}

 We now start at the output gate with the unique decomposition
 $(x,y)=(b,\vec{0})$ of vector $b$, and traverse an input-output path
 $P$ in the circuit \emph{backwards} by using the following rule: if
 $v$ is a currently reached gate, and $(x,y)$ is a decomposition at
 this gate, then go to that of the two inputs of $v$ which has a
 decomposition $(x',y')$ of norm $\nor{x'}\geq
 \tfrac{1}{2}\cdot\nor{x}$ (if both input gates have this property,
 then go to any of them). The claim above ensures that we will
 eventually reach some input node.

 If this input node holds the set $\{\vec{0}\}$, then the only
 decomposition $(x,y)=(\vec{0},b)$ of vector $b$ at this gate has norm
 $\nor{x}=\nor{\vec{0}}\leq 1$, and if this gate holds
 $\{\vec{e}_i\}$, then the only decomposition $(x,y)=(\vec{e}_i,
 b-\vec{e}_i)$ of $b$ at this gate has also  norm
 $\nor{x}=\nor{\vec{e}_i}\leq 1$. In both cases, we have that
 $\nor{x}\leq 1$, which is at most $\bbal\mm$, because $\bbal\geq
 1/\mm$.

 On the other hand, the (also unique) decomposition
 $(x,y)=(b,\vec{0})$ of the vector $b$ at the output gate has norm
 $\nor{x}=\nor{b}= \mm$, which is strictly larger than $\bbal \mm$,
 because $\bbal <1$.  So, there must be an edge $(u,v)$ in the path
 $P$ at which the jump from $\leq\! \bbal \mm$ to $>\! \bbal \mm$
 happens.  That is, there must be a decomposition $(x,y)$ at the gate
 $v$ and a decomposition $(x',y')$ at the gate $u$ such that $\nor{x}
 >\bbal \mm$ but $\nor{x'} \leq \bbal \mm$. By the above claim, we
 have $\nor{x'}\geq \tfrac{1}{2}\cdot\nor{x}$.  We have thus found a
 sumset $\pr{u}+\compl{u}$ and vectors $x'\in \pr{u}$ and
 $y'\in\compl{u}$ such that $x'+y'=b$ and $\tfrac{1}{2}\bbal \mm <
 \nor{x'}\leq \bbal \mm$, as desired.
\end{proof}

\subsection{Decomposition lemma for approximating $(\max,+)$ circuits}
In the following lemma, $0<\bal<1$ is an arbitrary fixed ``balance''
parameter. For a $0$-$1$ vector~$a$, let $|a|:=\skal{a,a}$ denote the
number of $1$s in~$a$.

\begin{lem}\label{lem:rect1}
  Let $\Phi$ be a $(\max,+)$ circuit of size $t$ approximating the
  maximization problem on a set $A\subseteq \{0,1\}^n$ within
  factor~$r\geq 1$, and let $B\subset\NN^n$ be the set of vectors
  produced by $\Phi$.  Then there exist $t$ or fewer sumsets
  $X+Y\subseteq B$ such that:
  \begin{enumerate}

  \item if $x\in X$ and $y\in Y$, then $x+y\leq a$ for some $a\in A$;

  \item $\skal{x,y}=0$ holds for all $x\in X$ and $y\in Y$;

  \item for every vector $a\in A$ with $|a|\geq r/\bal$ ones, at least
    one of these sumsets $X+Y$ contains vectors $x\in X$ and $y\in Y$
    such that:
    \begin{equation}\label{eq:bal-vectors}
      \skal{a,x+y}\geq \tfrac{1}{r}\cdot |a|\,,\ \skal{a,x} > \tfrac{\bal}{2r}\cdot|a| \ \mbox{ and }\ \skal{a,y}\geq
      \tfrac{1-\bal}{r}\cdot |a|\,.
    \end{equation}
  \end{enumerate}
\end{lem}

\begin{proof}
  By \cref{lem:maxA}, we know that the set $B$ has the following two
  properties:
  \begin{itemize}
  \item[\mbox{\rm (i)}] if $b\in B$, then $b\leq a$ for some $a\in A$;
  \item[\mbox{\rm (ii)}] if $a\in A$, then $\skal{a,b}\geq
    \frac{1}{r}\cdot|a|$ for some $b\in B$.
  \end{itemize}

  The Minkowski $(\cup,+)$ version $\Phi'$ of the circuit $\Phi$ has
  the same size $t$ and produces the same set~$B$.  When applied to
  the Minkowski circuit $\Phi'$, \cref{lem:decomp} gives us a
  collection of $t$ sumsets $X+Y\subseteq B$ with the following
  property holding for every norm-measure $\Nor:\NN^n\to\RR_+$, for
  every vector $b\in B$ of norm $\nor{b}>1$, and every real
  number~$\bbal$ satisfying $1/\nor{b}\leq \bbal<1$:
  \begin{enumerate}
  \item[$(\ast)$] \label{itm:decom1} at least one of the sumsets $X+Y$
    contains vectors $x\in X$ and $y\in Y$ such that $x+y=b$ and
    $\tfrac{\bbal}{2}\cdot\nor{b}<\nor{x}\leq \bbal\cdot \nor{b}$.
  \end{enumerate}

  Since the set $A$ consists of only $0$-$1$ vectors, property (i)
  implies that the set $B$ also consists of $0$-$1$ vectors. So,
  $X+Y\subseteq B$ implies that each of our sumsets $X+Y$ has the
  first two properties (1) and (2) claimed in \cref{lem:rect1}. It
  remains to show the third ``balancedness'' property~(3).

  Fix an arbitrary vector $a\in A$ with $|a|\geq r/\bal$ ones.
  Property (ii) of the set $B$ suggests to associate with $a$ the
  norm-measure $\nor{x}=\nnor{x}{a}:=\skal{a,x}$.  Then, by (ii),
  there is a vector $b\in B$ of norm $\nor{b}=\skal{a,b}\geq
  \mm:=|a|/r$, which is $>1$ since $|a|\geq r/\bal$, $r\geq 1$ and
  $\bal<1$. We also have $\bal\geq r/|a|\geq 1/\nor{b}$.  Hence, by
  the property $(\ast)$, at least one of our sumsets $X+Y$ contains
  vectors $x\in X$ and $y\in Y$ such that $x+y=b$ and the following
  inequalities hold:
  \begin{equation}\label{eq:rect}
    \tfrac{1}{2}\bal\cdot \mm < \nor{x}=\skal{a,x}\leq \bal\cdot  \mm\,.
  \end{equation}
  Now, the first inequality in \cref{eq:bal-vectors} follows since
  $\skal{a,x+y}=\skal{a,b}\geq \mm$, the second inequality in
  \cref{eq:bal-vectors} is the first inequality in \cref{eq:rect} and,
  since $\skal{x,y}=0$, the third inequality in \cref{eq:bal-vectors}
  follows from $\skal{a,x+y}\geq \mm$ and the second inequality in
  \cref{eq:rect}.
\end{proof}

\subsection{The rectangle bound}
\label{sec:set-theoretic}

A~\emph{rectangle} is a family of sets specified by giving a pair
$\A,\B$ of families of sets which is \emph{cross-disjoint} in that
$A\cap B=\emptyset$ holds for all sets $A\in \A$ and $B\in\B$.  The
rectangle $\R=\A\lor \B$ itself consists of all sets $A\cup B$ with
$A\in \A$ and $B\in\B$.  A~rectangle $\R$ \emph{lies below} a family
$\f$ if every set of $\R$ is contained in at least one set of~$\f$.

For real numbers $r\geq 1$ and $0<\bal<1$, we say that a set $F$
\emph{appears $(r,\bal)$-balanced} in a rectangle $\R=\A\lor\B$ if
there are sets $A\in\A$ and $B\in\B$ such that
\begin{equation}\label{eq:bal-sets}
  |F\cap (A\cup B)|\geq
  \tfrac{1}{r}\cdot |F|,\ \ |F\cap A| > \tfrac{\bal}{2r}\cdot |F|\ \mbox{ and }\ \ |F\cap B|\geq \tfrac{1-\bal}{r}\cdot |F|\,.
\end{equation}
Thus, a set $F$ appears balanced in a rectangle $\R$, if for some pair
of disjoint sets whose union belongs to $\R$, the union as well as the
sets themselves contain ``many enough'' elements of~$F$. In our
applications, we will only use the last two inequalities in
\cref{eq:bal-sets}, but the first inequality may also be important in
some applications.

The following lemma is just a translation of \cref{lem:rect1} from the
language of vectors to the language of sets. Recall that $\Max{\f}{r}$
denotes the minimum size of a $(\max,+)$ circuit approximating the
maximization problem on $\f$ within the factor~$r$.

\begin{lem}\label{lem:rect0}
  Let $r\geq 1$ and $0<\bal<1$.  If $\Max{\f}{r}\leq t$, then there
  are~$t$ or fewer rectangles lying below $\f$ such that every set of
  $\f$ with at least $r/\bal$ elements appears $(r,\bal)$-balanced in
  at least one of these rectangles.
\end{lem}

\begin{proof}
  Take a $(\max,+)$ circuit $\Phi$ of size $t=\Max{\f}{r}$
  approximating the maximization problem on $\f$ within factor~$r$.
  Let $A\subseteq\{0,1\}^n$ be the set of characteristic $0$-$1$
  vectors of the sets in $\f$. Hence, the circuit $r$-approximates the
  maximization problem on $A$.  Let $B\subset\NN^n$ be the set of
  vectors produced by $\Phi$.

  \Cref{lem:rect1} gives us $t$ or fewer sumsets $X+Y\subseteq B$ with
  properties (1)--(3) listed in this lemma.  By property (1), the set
  $B$ and, hence, each of our $t$ sumsets $X+Y$ consists of only
  $0$-$1$ vectors. So, each sumset $X+Y$ translates into the rectangle
  $\R=\A\lor\B$, where $\A$ is the family of supports
  $\supp{x}=\{i\colon x_i=1\}$ of vectors $x\in X$, and $\B$ is the
  family of supports $\supp{y}=\{i\colon y_i=1\}$ of vectors $y\in
  Y$. After this translation, property (1) of \cref{lem:rect1} implies
  that each of these rectangles $\R=\A\lor\B$ lies below our family
  $\f$, property (2) yields the cross-disjointness condition ($A\cap
  B=\emptyset$ for all sets $A\in \A$ and $B\in\B$), and property (3)
  implies that every set $F\in\f$ with $|F|\geq r/\bal$ elements
  appears $(r,\bal)$-balanced in at least one of these rectangles.
\end{proof}

In applications, we will use a direct consequence of \cref{lem:rect0}
formulated as a \emph{lower bound} on $\Max{\f}{r}$. Let $r\geq 1$ be
a given approximation factor, and $0<\bal <1$ any fixed real parameter
(we are free to choose this parameter).

\begin{thm}[Rectangle bound]\label{thm:rect}
  Let $\f$ be a family of feasible solutions, and $\h\subseteq \f$ some
  subfamily of feasible solutions, each of cardinality at least $r/\bal$. If for
  every rectangle $\R$ lying below $\f$, at most $h$ of the sets of
  $\h$ appear $(r,\bal)$-balanced, then $\Max{\f}{r}\geq |\h|/h$.
\end{thm}

In particular, \cref{res:rect} stated in \cref{sec:results}
corresponds to the balance parameter $\bal=2/3$. In the following
applications, we will always take $\f=\h$, but the possibility to
choose special subfamilies $\h\subseteq \f$ of feasible solutions may
be useful in other applications.  Note that, besides of the
cross-disjointness and balancedness of rectangles, the fact that the
rectangles must lie below $\f$ is also important. If, say, $\f$ is the
family of all perfect matchings in some graph, and if a rectangle $\R$
lies below $\f$, then we immediately know that every set of $\R$ must
be a union of two vertex-disjoint matchings.

\subsection{Maximization on designs}
\label{sec:designs}
We already know that there \emph{exist} many maximization problems for
which slight decrease of the allowed approximation factor from
$r=1+o(1)$ to $r=1$ can exponentially increase the size of $(\max,+)$
circuits (\cref{prop:counting1}). We also know \emph{explicit}
maximization problems for which such a jump in circuit size occurs
when decreasing the approximation factor from $r=2$ to $r=1$
(\cref{thm:sidon}).

Our goal in this section is to show that such jumps can happen for
arbitrarily \emph{large} approximation factors $r$: a slight decrease
of the allowed approximation factor $r$ can make tractable problems
intractable.  We demonstrate these jumps on maximization problems
whose families $\f$ of feasible solutions are ``combinatorial
designs'' (\cref{thm:hierarchy} below).

An $(m,d)$-\emph{design} (know also as a \emph{weak design}) is a
family $\f$ which is:
\begin{enumerate}
\item[$\circ$] $m$-\emph{uniform}: every set has exactly $m$ elements;
\item[$\circ$] $d$-\emph{disjoint}: no two distinct sets share $d$ or
  more elements in common.
\end{enumerate}

We will see soon (the upper bound in \cref{thm:hierarchy}) that, for
some $(m,d)$-designs~$\f\subseteq 2^{[n]}$ with $n=m^2$,
$\Max{\f}{r}=O(n)$ holds when the approximation factor $r=m/d$ is
allowed. On the other hand, we have the following general lower bound
on $\Max{\f}{r}$ for \emph{every} $(m,d)$ design $\f$, when the
allowed approximation factor $r$ is only slightly smaller than~$m/d$.

To state this bound, we need an auxiliary notation. For a family $\f$
of sets and a real number $l\geq 0$, let $\ddeg{\f}{l}$ denote the
maximal possible number of sets in $\f$ containing a fixed set with
$l$ (or more) elements. In other words, $\ddeg{\f}{l}$ is the maximal
possible number of sets in $\f$ whose intersection has $l$ (or more)
elements.  In particular, if $m$ is the maximum cardinality of a set
of~$\f$, then $|\f|=\ddeg{\f}{0}\geq \ddeg{\f}{1}\geq \ldots\geq
\ddeg{\f}{m}=1$, and $\ddeg{\f}{l}=0$ for all $l>m$.  Note that a
nonempty $m$-uniform family $\f$ is an $(m,d)$-design if and only if
$\ddeg{\f}{d}=1$. Also, $\ddeg{\f}{1}=1$ means that all sets of $\f$
are disjoint.

\begin{lem}\label{lem:gen-design}
  Let $\f$ be an $(m,d)$-design for $1\leq d <m$, $1/(d+1)\leq \bal <
  1$ and $l=\bal d/2$. For the factor $r=(1-\bal)m/d$, we have
  \[
  \Max{\f}{r}\geq \frac{|\f|}{\ddeg{\f}{l}}\,.
  \]
\end{lem}

\begin{proof}
  We are going to apply the rectangle bound (\cref{thm:rect}) with the
  balance parameter $\bal$. First, observe that $|F|\geq r/\bal$ holds
  for \emph{every} set $F\in\f$: since $|F|=m$, this is equivalent to
  the inequality $1\geq (1-\bal)/\bal d$, which holds because
  $\bal\geq 1/(d+1)$.

  Take an arbitrary rectangle $\R=\A\lor \B$ lying below $\f$. Let
  $\f_{\R}\subseteq \f$ be the family of all sets $F\in\f$ such that
  \[
  |F\cap A|\geq \tfrac{\bal}{2r}\cdot m = \tfrac{\bal}{2(1-\bal)}\cdot
  d\geq \tfrac{\bal}{2}\cdot d=l\ \ \mbox{ and }\ \ |F\cap B|\geq
  \tfrac{1-\bal}{r}\cdot m=d
  \]
  hold for some $A\in\A$ and $B\in\B$. By \cref{thm:rect}, it is
  enough to show that $|\f_{\R}|\leq \ddeg{\f}{l}$.  We can assume
  that all sets $B\in\B$ have $|B|\geq d$ elements: sets $B\in\B$ of
  size $|B|<d$ cannot fulfill $|F\cap B|\geq d$ and, hence, can be
  removed from $\B$ without changing~$\f_{\R}$.  Similarly, we can
  assume that all sets $A\in\A$ have $|A|\geq l$ elements.  Let $X$ be
  the union of all sets in $\A$; hence, $|X|\geq l$.

  \begin{clm}
    All sets of $\f_{\R}$ contain the set $X$.
  \end{clm}
  \begin{proof}
    Take a set $F\in\f_{\R}$. Then $|F\cap B|\geq d$ holds for some
    set $B\in\B$. On the other hand, since the rectangle $\R$ lies
    below $\f$, every set $A\cup B$ with $A\in\A$ must lie in some set
    of $\f$. Since all these sets contain the set $B$ with $|B|\geq d$
    elements, and since the family $\f$ is $d$-disjoint, this implies
    that \emph{all} sets of $\A\lor\{B\}$ and, hence, also the set
    $X\cup B$ must be contained in \emph{one} set $F_B$ of $\f$.
    Since both sets $F$ and $F_B$ of $\f$ contain the same set $F\cap
    B$ of size $|F\cap B|\geq d$, and since the family $\f$ is
    $d$-disjoint, the equality $F=F_B$ and, hence, the desired
    inclusion $X\subseteq F$ follows.
  \end{proof}

  Since the set $X$ has $|X|\geq l$ elements, the claim yields
  $|\f_{\R}|\leq \ddeg{\f}{|X|}\leq \ddeg{\f}{l}$ and, by
  \cref{thm:rect} (applied with $\h=\f$), the desired lower bound
  $\Max{\f}{r}\geq |\f|/|\f_{\R}|\geq |\f|/\ddeg{\f}{l}$ follows.
\end{proof}

\subsection{Factor hierarchy theorem}
We will now apply the general lower bound of \cref{lem:gen-design} to
explicit designs. Let $m$ be a prime power, $1\leq d\leq m$ an
integer, and consider the grid $\gf{m}\times\gf{m}$. The
\emph{polynomial $(m,d)$-design} $\f$ consists of all $|\f|=m^d$
subsets $S$ of points in this grid of the form $S=\{(a,p(a))\colon
a\in \gf{m}\}$ for a univariate polynomial $p=p(x)$ of degree at most
$d-1$ over $\gf{m}$.  Note that no two points of any of these sets $S$
lie in the same row of the grid.

The main combinatorial property of polynomial designs is the
following.
\begin{prop}\label{prop:pol-comb}
  Let $\f$ be a polynomial $(m,d)$-design, and $1\leq d\leq m$.  For
  every set of $l\leq d$ points of the grid $\gf{m}\times\gf{m}$, with
  no two in the same row, exactly $m^{d-l}$ sets of $\f$ contain this
  set. In particular, $\ddeg{\f}{l}\leq m^{d-l}$ holds for every
  $0\leq l\leq d$.
\end{prop}

\begin{proof}
  This is a direct consequence of a standard result in polynomial
  interpolation. For any $l\leq d$ distinct points
  $(a_1,b_1),\ldots,(a_l,b_l)$ in $\gf{m}\!\times\!  \gf{m}$, the
  number of polynomials $p(x)$ of degree at most $d-1$ satisfying
  $p(a_1)=b_1,\ldots,p(a_l)=b_l$ is either~$0$ (if $a_i=a_j$ holds for
  some $i\neq j$) or is exactly $m^{d-l}$: this latter number is exactly the number of solutions of the
  corresponding system of linear equations, with coefficients of $p$
  viewed as variables.
\end{proof}

We already know that $(\min,+)$ circuits approximating
\emph{minimization} problems on polynomial designs (within any finite
factor $r=r(m)$) must be large (\cref{cor:andreev}). Now we show that
also the \emph{maximization} problem on the polynomial $(m,d)$-design
is hard to approximate, as long as the allowed approximation factor is
smaller than~$m/d$.

\begin{thm}[Factor hierarchy theorem]\label{thm:hierarchy}
  Let $\f$ be a polynomial $(m,d)$-design for $1\leq d < m$, and
  $1/(d+1)\leq \bal < 1$. Then:
  \begin{enumerate}
  \item[\mbox{\rm (i)}] $\Max{\f}{r}\leq 3m^2$ for the factor $r=m/d$,
    but
  \item[\mbox{\rm (ii)}] $\Max{\f}{r}\geq m^{\bal d/2}$ for the factor
    $r=(1-\bal)m/d$.
  \end{enumerate}
\end{thm}
That is, the maximization problem on $\f$ can be approximated by a
small $(\max,+)$ circuit within the factor $r=m/d$, but the size of
approximating $(\max,+)$ circuits drastically increases when this
factor is only slightly decreased.

\begin{proof}
  To show the lower bound (ii), we just apply \cref{lem:gen-design},
  which yields the lower bound $\Max{\f}{r}\geq |\f|/\ddeg{\f}{l}$ for
  $l=\bal d/2$.  Since $|\f|=m^d$ and, by \cref{prop:pol-comb},
  $\ddeg{\f}{l}\leq m^{d-l}$, the desired lower bound $\Max{\f}{r}\geq
  m^l$ follows.

  So, it remains to show the upper bound~(i). Given an input weighting
  $x$ of the points of the grid $\gf{m}\times\gf{m}$, we can first use
  $m(m-1)$ $\max$ operations to compute $m$ numbers $y_1,\ldots,y_m$,
  where $y_i$ is the maximum weight of a point in the $i$th row of the
  grid.  We then apply the $(\max,+)$ circuit for the top $d$-of-$m$
  selection problem (see \cref{prop:dense}) to compute the sum $W$ of
  the \emph{largest} $d$ of the numbers $y_1,\ldots,y_m$ using at most
  $2dm$ additional $(\max,+)$ operations.  Hence, $W$ is a sum of
  weights of $d$ \emph{heaviest} points in the grid with no two in the
  same row: each $y_i$ picks only one point in the $i$th row. The main
  combinatorial property of polynomial designs (\cref{prop:pol-comb})
  implies that these $d$ points are contained in a (unique) set of
  $\f$. Hence, the found value $W$ cannot exceed the optimal value
  (the weights are nonnegative). On the other hand, the weight of $d$
  heaviest points of an optimal solution $F\in\f$ cannot exceed
  $W$. Since $|F|=m$, the weight of this solution cannot exceed
  $(m/d)W$, as desired.
\end{proof}

\begin{rem}\label{rem:greedy-design}
  The maximizing greedy algorithm also achieves the same approximation
  factor $m/d$ on the polynomial $(m,d)$-design: it will also first
  take the heaviest $d$ points of the grid $\gf{m}\times\gf{m}$, with
  no two lying in the same row. But this is already the best the
  greedy algorithm can do.

  To show this, take $\err>0$ arbitrarily small, and set
  $\scalar:=1/(1-\err/2)>1$. Take arbitrary two sets $A\neq B\in\f$,
  and a subset $S\subset A$ of $|S|=d$ elements. Since $\f$ is an
  $(m,d)$-design, $S$ cannot be contained in $B$.  So, give weight
  $\scalar>1$ to all elements of $S$, weight $1$ to all elements of
  $B\setminus S$, and zero weight to the rest. Then the maximizing
  (best-in) greedy algorithm picks elements of weight $\scalar$ first,
  gets all $|S|=d$ of them, but then is stuck because no element of
  weight $1$ fits; hence, the greedy algorithm achieves the total
  weight $\scalar|S|=\scalar d$. But the optimum is at least
  $|B|=m$. Hence, the approximation factor is at least $m/\scalar
  d=(1-\err/2)m/d >(1-\err)m/d$.
\end{rem}

\subsection{Greedy can beat approximating $(\max,+)$ circuits}
\label{sec:matching}
As \cref{rem:greedy-design} shows, \cref{thm:hierarchy} does
\emph{not} imply that the maximizing greedy algorithm can beat
approximating $(\max,+)$ circuits: small $(\max,+)$ circuits can also
achieve the greedy approximation factor on designs.

To show that the greedy algorithm can \emph{still} outperform
approximating $(\max,+)$ circuits, we consider another maximization
problem: maximum weight matchings in $k$-partite $k$-uniform
hypergraphs.  We have a set $V=V_1\cup \cdots\cup V_k$ of $|V|=mk$
vertices decomposed into $k$ disjoint blocks $V_1,\ldots,V_k$, each of
size $m$. Edges (called also \emph{hyperedges}) are $k$-tuples $e\in
V_1\times \cdots\times V_k$.  The ground set $E$ consists of all
$|E|=m^k$ edges.  Two edges are \emph{disjoint} if they differ in all
$k$ positions.  A \emph{matching} is a set of disjoint edges, and is a
\emph{perfect matching} if it has the maximum possible number $m$ of
edges.

The family $\f_{m,k}$ of feasible solutions of our problem consists of
all $|\f_{m,k}|=(m!)^{k-1}$ perfect matchings.  So, the maximization
problem on $\f_{m,k}$ is, given an assignment of nonnegative weights
$x_e$ to the edges $e\in E$, to compute the maximum total weight
\[
f(x)=\max\left\{x_{e_1}+\cdots+x_{e_m}\colon e_i\in E\,, \mbox{ and
    $e_i$ and $e_j$ are disjoint for all $i\neq j$ }\right\}
\]
of a perfect matching. Note that in the case $k=2$, $\f_{m,k}$
consists of perfect matchings in $K_{m,m}$, and the problem is to
compute the maximum weight of such a perfect matching.

The greedy algorithm can approximate the maximization problem on
$\f_{m,k}$ within the factor $k$ by just always picking the heaviest
of the remaining edges, untouched by the partial matching picked so
far. On the other hand, we have the following lower bound for
$(\max,+)$ circuits approximating this problem.

\begin{thm}\label{thm:matching}
  Let $m$ be a sufficiently large integer, and $k=k(m)$ be an integer
  such that $6\leq k\leq \log \sqrt{m}$. If $r\leq 2^k/9$, then
  \[
  \Max{\f_{m,k}}{r}=2^{\Omega(\sqrt{m})}\,.
  \]
\end{thm}

\begin{proof}
  We are going to apply the rectangle bound (\cref{thm:rect}) with
  $\bal:=2/3$; note that then both $\tfrac{\bal}{2}$ and $1-\bal$ in
  \cref{eq:bal-sets} are equal to $1/3$. So, take an arbitrary
  rectangle $\R=\A\lor \B$ lying below $\f=\f_{m,k}$. Hence, sets in
  $\A$ and in $\B$ are subsets of (hyper-)edges $e\in V_1\times
  \cdots\times V_k$.  Since $\R$ lies below our family $\f$, and $\f$
  consist of (perfect) matchings, all sets $A\cup B$ with $A\in\A$ and
  $B\in\B$ must also be matchings.  Take the integer $d:=\lceil
  m/3r\rceil$, and consider the family
  \[
  \f_{\R}=\{F\in\f\colon \mbox{$|F\cap A|\geq d$ and $|F\cap B|\geq d$
    for some $A\in\A$ and $B\in\B$ }\}\,.
  \]
  Our goal is to show a possibly small upper bound $|\f_{\R}|\leq h$
  on the number of sets in any such family. Then \cref{thm:rect} (for
  the balance parameter $\bal=2/3$) gives the lower bound
  $\Max{\f}{r}\geq |\f|/h$.

  Since the rectangle $\R=\A\lor \B$ is cross-disjoint, we know that
  the matchings $A\in\A$ and $B\in\B$ must be edge-disjoint, that is,
  $A\cap B=\emptyset$ must hold.  However, since the sets $A\cup B$
  are also matchings ($\R$ lies below $\f$), we actually know that
  matchings $A$ and $B$ are even \emph{vertex-disjoint}: if
  $S\subseteq V$ is the set of vertices belonging to at least one edge
  of a matching in $\A$, and $T\subseteq V$ is the set of vertices
  belonging to at least one edge of a matching in $\B$, then $S\cap
  T=\emptyset$ (this is a crucial property).  Note that in the proof
  of the lower bound for designs (\cref{lem:gen-design}), the
  cross-disjointness property of rectangles was \emph{not} used.

  So, call a matching $A\subseteq V_1\times\cdots\times V_k$ an
  $S$-\emph{matching} if $A\subseteq (V_1\cap S)\times\cdots\times
  (V_k\cap S)$ holds, that is, if edges of $A$ only match vertices of
  $S$; $T$-\emph{matchings} are defined similarly. By the definition
  of $\f_{\R}$, every perfect matching $F\in\f_{\R}$ has at least $d$
  edges lying in $(V_1\cap S)\times\cdots\times (V_k\cap S)$, and at
  least $d$ edges lying in $(V_1\cap T)\times\cdots\times (V_k\cap
  T)$. In particular, every perfect matching $F\in\f_{\R}$ must
  contain at least one matching $A\cup B$, where $A$ is an
  $S$-matching with $|A|=d$ edges and $B$ is a $T$-matching with
  $|B|=d$ edges. It therefore suffices to upper-bound the number of
  perfect matchings~$F$ with this property.

  We can pick any such pair $(A,B)$ as follows.  Let $S_i=S\cap V_i$
  and $T_i=T\cap V_i$ for $i=1,\ldots,k$. We can assume that each of
  these $2k$ sets has at least $d$ vertices, for otherwise none of the
  $S$-matchings or of the $T$-matchings could have $\geq d$ edges,
  implying that $\f_{\R}=\emptyset$.

  \begin{enumerate}
  \item[$\circ$] Pick in each $S_i$ a subset $S_i'\subseteq S_i$ of
    $|S_i'|=d$ vertices, and in each $T_i$ a subset $T_i'\subseteq
    T_i$ of $|T_i'|=d$ vertices.  There are at most
    \[
    \prod_{i=1}^k\binom{m_i}{d}\binom{m-m_i}{d} \leq \binom{m}{2d}^k
    \]
    possibilities to do this, where $m_i=|S_i|$.

  \item[$\circ$] Pick a perfect matching $A$ in $S_1'\times \cdots
    \times S_k'$ and a perfect matching $B$ in $T_1'\times \cdots
    \times T_k'$. There are only
    $\left[(d!)^{k-1}\right]^2=(d!)^{2(k-1)}$ possibilities to do
    this.
  \end{enumerate}

  After a pair $(A,B)$ of matchings is picked, there are at most
  $[(m-2d)!]^{k-1}$ possibilities to extend $A\cup B$ to a perfect
  matching.  Thus,
  \[
  |\f_{\R}| \leq \binom{m}{2d}^{k} (d!)^{2(k-1)} [(m-2d)!]^{k-1}
  = \binom{m}{2d}\left[m!\cdot \binom{2d}{d}^{-1} \right]^{k-1}\,,
  \]
  where the equality follows because $\binom{m}{2d}=m!/(2d)!(m-2d)!$
  and $(2d)!/(d!)^2=\binom{2d}{d}$.  Since there are $|\f|=(m!)^{k-1}$
  perfect matchings, the rectangle bound (\cref{thm:rect}) yields the
  following lower bound on $t=\Max{\f}{r}$:
  \begin{align*}
    t& \geq \frac{|\f|}{|\f_{\R}|} \geq
    \frac{\binom{2d}{d}^{k-1}}{\binom{m}{2d}} \geq
    \left(\frac{2^{2d}}{d}\right)^{k-1}\cdot \left(\frac{2d}{\euler
        m}\right)^{2d}=\frac{1}{d^{k-1}}\left(\frac{2^kd}{\euler
        m}\right)^{2d} \geq \frac{1}{d^{k-1}}\left(\frac{2^k}{3\euler
        r}\right)^{2d}\,,
  \end{align*}
  where the second inequality follows from the inequalities
  $\binom{m}{2d}\leq (\euler m/2d)^{2d}$ and $\binom{2d}{d}\geq
  2^{2d}/\sqrt{4d}\geq 2^{2d}/d$, and the last inequality follows
  because (by our choice) $d=\lceil m/3r\rceil\geq m/3r$.  Our
  approximation factor is $r= 2^k/9$. Since clearly $d\leq m$, we have
  a lower bound
  \[
  t\geq \left(\frac{3}{\euler}\right)^{6m/2^k}\cdot d^{-k} \geq 2^{0.8
    m/2^k - k\log m }\,.
  \]
  From our assumption $k\leq \log \sqrt{m}$, we have $m/2^k\geq
  \sqrt{m}\gg k\log m$, and the desired lower bound $t\geq
  2^{\Omega(m/2^k)}\geq 2^{\Omega(\sqrt{m})}$ follows.
\end{proof}

\section{What do approximating tropical circuits produce?}
\label{sec:structure-tight}
If we know that a tropical circuit \emph{approximates} a given
optimization (minimization or maximization) problem within a given
factor, what can then be said about the set of vectors \emph{produced}
by that circuit?  Using \emph{elementary arguments}, we partially
answered this question in \cref{lem:minA,lem:maxA}: we gave
properties, which the sets of produced vectors must \emph{necessarily}
have (these properties were already sufficient for our purposes). We
will now use convexity arguments to give properties of produced sets
that are also \emph{sufficient} for circuits to approximate given
problems.

\subsection{A version of Farkas' lemma}
Recall that a vector $c\in\RR^n$ is a \emph{convex combination} (or a
\emph{weighted average}) of vectors\footnote{We use the arrow notation
  $\vvec{b}{i}$ for vectors only when they are indexed.}
$\vvec{b}{1},\ldots,\vvec{b}{m}$ in $\RR^n$ if there are real scalars
$\lambda_1,\ldots,\lambda_m\geq 0$ such that
\[
\lambda_1+\cdots+\lambda_m=1\ \ \mbox{ and }\ \ c=
\lambda_1\cdot\vvec{b}{1}+\cdots+\lambda_m\cdot\vvec{b}{m}\,.
\]
It is easy to see the following \emph{averaging property}: for every
vector $x\in\RR^n$ and every convex combination $c$ of vectors in $B$,
we have $\min_{b\in B}\skal{b,x}\leq \skal{c,x}\leq \max_{b\in B}\
\skal{b,x}$.

We will need the following formulation of Farkas' lemma due to
Fan~\cite[Theorem~4]{Fan56}, see also
\cite[Corollary~7.1h]{schrijver}.

\begin{lem}[Farkas' lemma~\cite{Fan56}]\label{lem:fan}
  Let $u,\vvec{u}{1},\ldots,\vvec{u}{m}\in\RR^n$, and
  $\alpha,\alpha_1,\ldots,\alpha_m\in\RR$. The following two assertions are
  equivalent.
  \begin{enumerate}
  \item $\forall y\in \RR^n$ inequalities $\skal{\vvec{u}{1},y}\geq
    \alpha_1, \ldots,\skal{\vvec{u}{m},y}\geq \alpha_m$ imply $\skal{u,y}\geq
    \alpha$.

  \item $\exists\ \lambda_1,\ldots,\lambda_m\in\RR_+$ such that
    $u=\sum_i\lambda_i \vvec{u}{i}$ and $\alpha\leq \sum_i\lambda_i \alpha_i$.
  \end{enumerate}
\end{lem}

This relates optimization with convex combinations.

\begin{lem}\label{lem:JS}
  For any vectors $a,\vvec{a}{1},\ldots,\vvec{a}{m}\in\RR^n$ the
  following two assertions are equivalent.
  \begin{enumerate}
  \item $\forall x\in \RR_+^n\colon\skal{a,x}\geq
    \min_i\skal{\vvec{a}{i},x}$.

  \item $\exists\ \lambda_1,\ldots,\lambda_m\in\RR_+\colon\
    \sum_i\lambda_i=1$ and $a\geq \sum_i\lambda_i\vvec{a}{i}$.
  \end{enumerate}
\end{lem}

\begin{proof}
  The implication (2) $\Rightarrow$ (1) follows directly from the
  aforementioned averaging property of convex combinations. To show
  the converse implication (1) $\Rightarrow$ (2), observe that the
  assertion (1) is equivalent to the assertion that the set of
  inequalities $\skal{\vvec{a}{i},x}\geq z$ and $\skal{\vvec{e}{j},x}
  \geq 0$ for $i=1,\ldots,m$ and $j=1,\ldots,n$ implies the inequality
  $\skal{a,x}\geq z$. We use the inequalities $\skal{\vvec{e}{j},x}\geq 0 $
  to ensure that we only consider vectors $x$ in $\RR_+^n$ (with
  no negative entries).

  By taking $y=(x,z)$, $u=(a,-1)$, $\vvec{u}{i}=(\vvec{a}{i},-1)$ for
  $i=1,\ldots,m$, and $\vvec{u}{m+j}=(\vvec{e}{j},0)$ for
  $j=1,\ldots,n$, the above assertion turns into the assertion that
  for every vector $y$ in $\RR^{n+1}$, the system of inequalities
  $\skal{\vvec{u}{i},y}\geq 0$ for $i=1,\ldots,m+n$, implies the
  inequality $\skal{u,y}\geq 0$. Then, by \cref{lem:fan}, there exist
  $\lambda_1,\ldots,\lambda_{m+n}\in\RR_+$ such that
  \[
  (a,-1)=\sum_{i=1}^m\lambda_i(\vvec{a}{i},-1)
  +\sum_{j=1}^n\lambda_{m+j}(\vvec{e}{j},0)\,.
  \]
  This yields $\lambda_1+\cdots+\lambda_m=1$ and $a\geq
  \sum_i\lambda_i\vvec{a}{i}$, as desired.
\end{proof}

The following direct consequence of \cref{lem:JS} compares the values
of optimization problems. For a set $U$ of real vectors, let
$\conv{U}$ denote the set of all convex combinations of vectors in
$U$, that is, the convex hull of~$U$. Say that a set $U\subseteq
\RR^n$ \emph{lies above} a set $V\subseteq\RR^n$ if $\forall u\in U\
\exists v\in V\colon\ u\geq v$, and that $U$ \emph{lies below} $V$ if
$\forall u\in U\ \exists v\in V\colon\ u\leq v$.

\begin{lem}\label{lem:farkas1}
  Let $U,V\subset\RR^{n}$ be finite sets of vectors. Then
  \begin{enumerate}
  \item $\forall x\in\RR_+^n\colon \min_{u\in U}\ \skal{u,x}\geq
    \min_{v\in V}\skal{v,x}$ if and only if $U$ lies above $\conv{V}$;
  \item $\forall x\in\RR_+^n\colon \max_{u\in U}\ \skal{u,x} \leq
    \max_{v\in V}\skal{v,x}$ if and only if $U$ lies below $\conv{V}$.
  \end{enumerate}
\end{lem}
Claim (1) follows directly from~\cref{lem:JS}. The second claim (2)
also follows from~\cref{lem:JS} by using the equality
$\max(x,y)=-\min(-x,-y)$.

\subsection{Consequences for tropical circuits}
\label{sec:exact}
Recall that the maximization (resp., minimization) problem on a given
set $A\subset\NN^n$ of feasible solutions is, for every input
weighting $x\in\RR_+^n$, to compute the maximum (resp., minimum)
weight $\skal{a,x}=a_1x_1+\cdots+a_nx_n$ of a feasible solution $a\in
A$.

The consequence of Farkas' lemma (\cref{lem:farkas1}) directly yields
the following complete characterization of the properties of sets of
vectors produced by approximating $(\max,+)$ circuits.

\begin{lem}[Maximization]\label{lem:max}
  Let $A\subset\NN^n$ be some finite set of vectors, $\Phi$ be a
  $(\max,+)$ circuit, and $B\subset\NN^n$ the set of vectors produced
  by $\Phi$. Then the following two assertions are equivalent.

  \begin{enumerate}
  \item $\Phi$ approximates the maximization problem on $A$ within a
    factor~$r$.

  \item $B$ lies below $\conv{A}$ and $\tfrac{1}{r}\cdot A$ lies below
    $\conv{B}$.
  \end{enumerate}
\end{lem}

\begin{proof}
  By \cref{lem:constant-free}, we can assume the circuit $\Phi$ is
  constant-free. Hence, by \cref{prop:function}, the circuit solves
  the maximization problem of the form $\Phi(x)=\max_{b\in B}\
  \skal{b,x}$.  The maximization problem on $A$ is of the form $f(x)
  =\max_{a\in A}\skal{a,x}$. The circuit approximates the maximization
  problem on $A$ within factor~$r$ if and only if $\tfrac{1}{r}\cdot
  f(x)\leq \Phi(x)\leq f(x)$ holds for all weightings $x\in\RR_+^n$.

  When applied with $U=B$ and $V=A$, \cref{lem:farkas1}(2) implies
  that the inequality $\Phi(x)\leq f(x)$ holds if and only if $B$ lies
  below $\conv{A}$.  When applied with $U=\tfrac{1}{r}\cdot A$ and
  $V=B$, this lemma implies that the inequality $\tfrac{1}{r}\cdot
  f(x) \leq \Phi(x)$ holds if and only if $\tfrac{1}{r}\cdot A$ lies
  below $\conv{B}$.
\end{proof}

We say that a set $U\subseteq \RR^n$ \emph{lies tightly above} the
convex hull $\conv{V}$ of a set $V\subseteq\RR^n$ if for every vector
$u\in U$, the inequality $u\geq c$ holds for some convex combination
$c$ of vectors $v\in V$ with the \emph{same} support as that of
$u$. That is, we now additionally have that \emph{none} of the vectors
$v\in V$ in the convex combination $c$ has a zero in a position $i$
where $u_i\neq 0$.

\begin{lem}[Minimization]\label{lem:min}
  Let $A\subset\NN^n$ be some finite set of vectors, $\Phi$ be a
  $(\min,+)$ circuit, and $B\subset\NN^n$ the set of vectors produced
  by $\Phi$.

  Then the following two assertions are equivalent.

  \begin{enumerate}
  \item $\Phi$ approximates the minimization problem on $A$ within a
    factor~$r$.

  \item $B$ lies above $\conv{A}$ and $r\cdot A$ lies above
    $\conv{B}$.
  \end{enumerate}
  If $A\subseteq\{0,1\}^n$ and $A$ is an antichain, then these
  assertions are equivalent to:
  \begin{enumerate}
  \item[\mbox{\rm 3.}] $B$ lies above $A$ and $r\cdot A$ lies tightly
    above $\conv{B}$.
  \end{enumerate}
\end{lem}

\begin{proof}
  By \cref{lem:constant-free}, we can assume the circuit $\Phi$ is
  constant-free. Hence, by \cref{prop:function}, the circuit solves
  the minimization problem of the form $\Phi(x)=\min_{b\in B}\
  \skal{b,x}$.  The minimization problem on $A$ is of the form $f(x)
  =\min_{a\in A}\skal{a,x}$. The circuit approximates the minimization
  problem on $A$ within factor~$r$ if and only if $f(x)\leq
  \Phi(x)\leq r\cdot f(x)$ holds for all weightings $x\in\RR_+^n$.
  When applied with $U=B$ and $V=A$, \cref{lem:farkas1}(1) implies
  that the inequality $\Phi(x)\geq f(x)$ holds if and only if $B$ lies
  above $\conv{A}$.  When applied with $U=r\cdot A$ and $V=B$, this
  lemma implies that the inequality $r\cdot f(x) \geq \Phi(x)$ holds
  if and only if $r\cdot A$ lies above $\conv{B}$.  This shows the
  equivalence of the assertions (1) and (2).

  Suppose now that $A\subseteq\{0,1\}^n$, and that $A$ is an
  antichain. The implication (3) $\Rightarrow$ (2) is obvious. So,
  assume that the set $B$ has property~(2), i.e., that $B$ lies above
  $\conv{A}$ and $r\cdot A$ lies above $\conv{B}$.

  To show that $B$ lies above the set $A$ (not only above its convex
  hull), take an arbitrary vector $b\in B$. Since $B$ lies above
  $\conv{A}$, there must be a vector $a\in A$ and a scalar
  $0<\lambda\leq 1$ such that $b\geq \lambda\cdot a$. Since $a$ is a
  $0$-$1$ vector, and $b$ is a nonnegative \emph{integer} vector,
  $b\geq a$ must hold.

  To show that $r\cdot A$ lies tightly above $\conv{B}$, take an
  arbitrary vector $a\in A$.  Since, by (2), the set $r\cdot A$ lies
  above $\conv{B}$, the inequality $ r\cdot a\geq c$ must hold for
  some convex combination $c=\lambda_1\cdot
  \vvec{b}{1}+\cdots+\lambda_m\cdot \vvec{b}{m}$ of vectors in $B$,
  where all scalars $\lambda_i$ are positive.  It remains to show that
  then $\supp{\vvec{b}{i}}=\supp{a}$ holds for all $i=1,\ldots,m$,
  i.e., that all the vectors in this convex combination have the same
  support as our vector~$a$.

  Since (as we have just shown) the set $B$ lies above the set $A$,
  there must be (not necessarily distinct) vectors
  $\vvec{a}{1},\ldots,\vvec{a}{m}$ in $A$ such that $\vvec{b}{i}\geq
  \vvec{a}{i}$ for all $i=1,\ldots,m$ and, hence,
  $c=\sum_{i=1}^m\lambda_i\cdot \vvec{b}{i}\geq
  \sum_{i=1}^m\lambda_i\cdot \vvec{a}{i}$. The inequality $r\cdot
  a\geq c$ implies that $\supp{a} \supseteq
  \supp{\vvec{b}{i}}\supseteq \supp{\vvec{a}{i}}$ must hold for all
  $i$. Since $A$ is an antichain and consists of only $0$-$1$ vectors,
  this implies $\vvec{a}{i}=a$ for all $i$. We thus have
  $\supp{\vvec{b}{i}}=\supp{a}$ for all $i=1,\ldots,m$, as desired.
\end{proof}

\begin{rem}[Approximation using arithmetic circuits]\label{rem:min-to-arithm}
  \Cref{lem:min} implies that in order to show that the minimization
  problem on an antichain $\f\subseteq 2^{[n]}$ can be
  $r$-approximated by a $(\min,+)$ circuit of size $t$, it is enough
  to design a monotone \emph{arithmetic} $(+,\times)$ circuit $\Phi$
  of size $\leq t$ such that the polynomial computed by this circuit
  has the following two properties:
  \begin{enumerate}
  \item for every monomial $\prod_{i\in T}x_i^{d_i}$ there is a set
    $S\in\f$ with $S\subseteq T$;
  \item for every set $S\in\f$ there is a monomial $\prod_{i\in
      T}x_i^{d_i}$ with $T=S$ and all $d_i\leq r$.
  \end{enumerate}
  Indeed, property (1) ensures that the set $B$ of vectors produced by
  the arithmetic circuit $\Phi$ lies above the set $A$ of
  characteristic $0$-$1$ vectors of sets in $\f$, while property (2)
  ensures that the set $r\cdot A$ lies tightly above $\conv{B}$. By
  \cref{lem:min}, the $(\min,+)$ version of $\Phi$ $r$-approximates
  the minimization problem on~$\f$.
\end{rem}

\section{A tight boolean bound for $(\min,+)$ circuits}
\label{sec:bool-tight}
In \cref{sec:bool-bound}, we have shown (\cref{thm:bool}) that the
monotone boolean circuit complexity of the decision versions of
minimization problems is a \emph{lower bound} on the size of
$(\min,+)$ circuits approximating these problems within \emph{any}
finite approximation factor $r\geq 1$.  \Cref{lem:min} will allow us
to take the factor $r$ into account, that is, to show that
approximating $(\min,+)$ circuits and monotone boolean circuits are
even more tightly related (as given in \cref{thm:minexact} below). For
this purpose, we introduce the concept of ``semantic degree'' of
monotone boolean circuits.

\subsection{Semantic degree of boolean circuits}
\label{sec:sem-deg}

A~\emph{minterm} of a monotone boolean function $f(x_1,\ldots,x_n)$ is
a vector $a\in\{0,1\}^n$ such that $f(a)=1$, but $f(a')=0$ for any
vector $a'$ obtained by switching any single $1$-entry of $a$ to~$0$. The
boolean function \emph{defined} by a finite set $A\subset\NN^n$ of
vectors is of the form
\[
f_A(x)= \bigvee_{a\in A}~\bigwedge_{i\in \supp{a}}x_i\,,
\]
where, as before, $\supp{a}=\{i\colon a_i\neq 0\}$ is the support of
vector~$a$. In particular, if $A\subset\{0,1\}^n$ is the set of
minterms of a boolean function, then this function is of the form
$f_A$ (is defined by the set of its minterms).

A monotone boolean $(\lor,\land)$ circuit $\bphi$ for a boolean
function $f_A$ (defined by its set $A$ of minterms) not only
\emph{computes} the function $f_A$, but also \emph{produces} (purely
syntactically) some finite set $B\subset\NN^n$ of vectors, as given in
\cref{sec:produced}. By \cref{prop:function}, the circuit $\bphi$
computes the boolean function $f_B$ defined by the set $B$.  Since the
circuit $\bphi$ \emph{computes} the function $f_A$, we know that
$f_B(x)=f_A(x)$ must hold for all $x\in\{0,1\}^n$.  The ``semantic
degree'' of the circuit $\bphi$ (motivated by \cref{lem:min}) gives an
upper bound on the magnitudes of entries of particular (not all)
vectors of the set~$B$.

Namely, we define the \emph{semantic degree}, $\deg{\phi}$, of $\bphi$
as the minimum real number~$r$ such that the set $r\cdot A$ lies
tightly above the convex hull $\conv{B}$ of the set $B$ produced by
the circuit~$\bphi$. Recall that this means that for every minterm
$a\in A$ there are vectors $\vvec{b}{1},\ldots,\vvec{b}{m}\in B$ and
positive scalars $\lambda_1,\ldots,\lambda_m$ such that
$\lambda_1+\cdots+\lambda_m=1$,
$\supp{\vvec{b}{1}}=\ldots=\supp{\vvec{b}{m}}=\supp{a}$ and
\begin{equation}\label{eq:semantic}
  a\leq \lambda_1\cdot
  \vvec{b}{1}+\cdots+\lambda_m\cdot \vvec{b}{m} \leq r\cdot a\,.
\end{equation}

We use the adjective ``semantic'' because $\deg{\bphi}$ depends on the
\emph{function} computed by $\bphi$, that is, on the set $A$ of
minterms of this function.  Note that the first inequality in
\cref{eq:semantic} always holds  because $a$ is a $0$-$1$ vector, and
$\supp{\vvec{b}{i}}=\supp{a}$ holds for all vectors~$b_i$ (we included this inequality just for clarity of the concept).

\subsection{The converse of the boolean bound
  (Theorem~\ref{thm:bool})}

For a finite set $A\subset\NN^n$, let $\BBool{A}{r}$ denote the
minimum size of a monotone boolean circuit of semantic degree at most
$r$ computing the boolean function $\bool{A}$ defined by~$A$.

\begin{thm}[Tight boolean bound]\label{thm:minexact}
  If $A\subset\{0,1\}^n$ is an antichain, then $
  \Min{A}{r}=\BBool{A}{r}$ holds for every $r\geq 1$.
\end{thm}

\begin{proof}
  To show $\Min{A}{r}\leq \BBool{A}{r}$, take a monotone boolean
  $(\lor,\land)$ circuit $\bphi$ of semantic degree $r$ computing the
  boolean function $\bool{A}$ defined by $A$. We can assume that the
  circuit is constant-free: $0$ and $1$ are the only possible
  constants, and they can be trivially eliminated from the
  circuit. Let $B\subset\NN^n$ be the set of vectors produced by the
  circuit $\bphi$.  By \cref{prop:function}, the circuit $\bphi$
  computes the boolean function $f_B$ defined by this set, that is,
  $\bool{B}(x)=\bool{A}(x)$ holds for all inputs $x\in\{0,1\}^n$.

  Let $\Phi$ be the (also constant-free) tropical $(\min,+)$ version
  of the boolean circuit $\bphi$ obtained from $\bphi$ by replacing
  each $\lor$-gate by a $\min$-gate, and each $\land$-gate by a
  $+$-gate.  The $(\min,+)$ circuit $\Phi$ produces the same set~$B$.

  Since $\bool{B}(x)\leq \bool{A}(x)$ must hold for all inputs
  $x\in\{0,1\}^n$, for every $b\in B$ there must be a vector $a\in A$
  with $\supp{b}\supseteq \supp{a}$. Since vectors in $A$ are $0$-$1$
  vectors, this latter inclusion yields $b\geq a$. Thus, the set $B$
  lies above $A$ and, hence, also above $\conv{A}$. Since the circuit
  $\bphi$ has semantic degree $r$, we additionally have that the set
  $r\cdot A$ lies above $\conv{B}$.  By \cref{lem:min}, the circuit
  $\Phi$ approximates the minimization problem on~$A$ within the
  factor~$r$.

  To show $\BBool{A}{r}\leq \Min{A}{r}$, take a tropical $(\min,+)$
  circuit $\Phi$ approximating the minimization problem on $A$ within
  the factor $r$, and let $B\subset\NN^n$ be the set of vectors
  produced by $\Phi$. By \cref{lem:constant-free}, we can assume that
  the circuit $\Phi$ is constant-free.  Let $\bphi$ be the boolean
  version of the tropical circuit $\Phi$ obtained by replacing each
  $\min$-gate by an $\lor$-gate, and each $+$-gate by an
  $\land$-gate. (Recall that $\min$ and $\lor$ are ``additions'' and $+$ and $\land$ are ``multiplications'' in the corresponding semirings.)
  The circuit $\bphi$ produces the same set~$B$.
  By \cref{lem:min}, we know that the set $B$ has the following two
  properties:
  \begin{enumerate}
  \item[(i)] $B$ lies above $A$ ;
  \item[(ii)] $r\cdot A$ lies tightly above $\conv{B}$.
  \end{enumerate}
  By property (ii), the semantic degree of the boolean circuit $\bphi$
  is at most~$r$.  On the other hand, property (i) implies that the
  support of every vector $b\in B$ contains the support of at least
  one vector $a\in A$, and property (ii) implies that the support of
  every vector $a\in A$ contains the support of at least one vector
  $b\in B$.  In terms of \cref{prop:bool-str}, this means that the set
  $B$ is similar to the set $A$. Since (by \cref{prop:function}) the
  circuit computes a boolean function $f_B$ defined by the set $B$,
  \cref{prop:bool-str} itself implies that $\bphi$ computes the
  boolean function defined by the set $A$.
\end{proof}

\subsection{Bounds on semantic degree}
\label{sec:sem-deg-upper}

An $r$-\emph{bounded copy} of a boolean vector $a\in\{0,1\}^n$ is an
integer vector $b\in\NN^n$ which has the same nonzero positions as
$a$, and every nonzero position of $b$ is at most~$r$.
In particular, the unique $1$-bounded copy of $a$ is the vector $a$ itself.
Recall that
$|a|=\skal{a,a}$ is the number of ones in a $0$-$1$ vector~$a$.

\begin{prop}\label{prop:deg-upper}
  Let $\bphi$ be a monotone boolean circuit computing a boolean
  function $f$, $A\subset\{0,1\}^n$ the set of minterms of~$f$, and
  $B\subset\NN^n$ the set of vectors produced by the circuit
  $\bphi$. Then the following holds.
  \begin{enumerate}
  \item $\deg{\bphi}=1$ if and only if $A\subseteq B$.
  \item $\deg{\bphi}\leq r$ holds if for every $a\in A$ the set $B$
    contains at least one $r$-bounded copy of$~a$.
  \item If $\deg{\bphi}\leq r$ holds, then for every $a\in A$ the set
    $B$ contains at least one $s$-bounded copy of$~a$ for $s\leq
    r|a|-|a|+r$.
  \end{enumerate}
\end{prop}
Note that (1) is a special case of (2) and (3) for $r=1$
\begin{proof}
  The ``if'' direction in (1) is obvious. The ``only if'' direction
  follows from a simple observation: a convex combination $\lambda_1
  b_1+\cdots+\lambda_m b_m$ of positive \emph{integers} $b_i$ with all
  $\lambda_i>0$ can be equal $1$ only if $b_1=\ldots=b_m=1$.

  Claim (2) is also obvious, because for every $0$-$1$ vector $a$, the
  inequality $b\leq r\cdot a$ holds for \emph{every} $r$-bounded copy
  $b$ of~$a$.

  To show claim (3), assume that $\deg{\bphi}\leq r$. Take any minterm
  $a\in A$ and let $m=|a|$ be the number of ones in $a$.  By the
  definition of the semantic degree, we know that there must be a
  convex combination $c=\sum_{i=1}^l\lambda_i\cdot\vvec{b}{i}$
  of vectors $b_i\in B$ with all supports
  $\supp{\vvec{b}{i}}=\supp{a}$ such that $c\leq r\cdot a$ holds.  By
  Carath\'eodory's theorem~\cite{caratheodory}, if a vector is in the
  convex hull of some set $P\subseteq \RR^m$ of vectors, then this
  vector can be written as a convex combination of $m+1$ or fewer
  vectors in $P$. So, by taking
  $P=\{\vvec{b}{1},\ldots,\vvec{b}{l}\}$, we can assume that $l\leq
  |a|+1=m+1$.

  Consider the vectors $\Vvec{b}{i}:=\vvec{b}{i}-a\geq \vnul$ (the vectors $\Vvec{b}{i}$ are nonnegative, because vectors $\vvec{b}{i}$ have the same support as $a$). Then
  $c=a + c'$ with $c':=\sum_{i=1}^l \lambda_i \Vvec{b}{i} = c-a \leq
  r\cdot a -a = (r-1)\cdot a$.  Since
  $\lambda_1+\cdots+\lambda_l=1$, there must be an $i$ such that
  $\lambda_i\geq 1/l\geq 1/(m+1)$. From $\lambda_i\cdot
  \Vvec{b}{i}\leq c'\leq (r-1)\cdot a$, and since $a$ is a $0$-$1$
  vector, we have that all entries of vector $\Vvec{b}{i}$ must be at
  most $(r-1)/\lambda_i\leq (r-1)(m+1)$. Hence, all entries of the
  vector $\vvec{b}{i}$ are at most $(r-1)(m+1)+1=rm-m+r$, as desired.
\end{proof}

\begin{rem}
An apparent advantage of  \cref{prop:deg-upper} is that it avoids the somewhat involved definition of
  the semantic degree via convex hulls. Items (i) and (ii) may be
  useful when proving \emph{upper} bounds, while items (i) and (iii)
  may be useful when proving \emph{lower} bounds on the size of
  monotone boolean circuits of bounded semantic degree.
\end{rem}

\begin{rem}
  Note that the upper bound $s\leq r|a|-|a|+r$ in item (iii) of
  \cref{prop:deg-upper} cannot be substantially improved. Take $m=|a|$
  vectors $\vvec{b}{i}:=a+m(r-1)\vvec{e}{i}$, and let all
  $\lambda_i:=1/m$. Then the convex combination
  $c=\sum_{i=1}^m\lambda_i\cdot\vvec{b}{i}=a+(r-1)\cdot a=r\cdot a$
  satisfies $c\leq r\cdot a$, but every vector $\vvec{b}{i}$ in this combination
   has $s=1+m(r-1) = r|a|-|a|+1$ as one of it entries.
\end{rem}

The following example shows that the semantic degree of monotone
boolean circuits can be small even when some vectors produced by the circuit
have very large entries.

\begin{ex}[Shortest paths]\label{ex:BF}
  Let $A$ be the set of characteristic $0$-$1$ vectors of all simple
  paths in $K_n$ between two fixed vertices $s$ and $t$.  Then the
  boolean function $\bool{A}$ defined by $A$ is the $s$-$t$
  connectivity function STCONN on $n$-vertex graphs.  The
  Bellman--Ford pure DP algorithm for the shortest $s$-$t$ path
  problem gives us a monotone boolean $(\lor,\land)$ circuit $\bphi$
  of size $O(n^3)$ computing the boolean function $f_{A}$.  The
  circuit has gates $\BF{l}{j}$ at which the existence of a path from
  vertex $s$ to vertex $j$ with at most $l$ edges is detected. Then
  $\BF{1}{j}=x_{s,j}$ for all $j\neq s$, and the recursion of
  Bellman--Ford is to compute $\BF{l+1}{j}$ as the OR of $\BF{l}{j}$
  and all $\BF{l}{i}\land x_{i,j}$ for $i\not\in\{s,j\}$.  The output
  gate is $\BF{n-1}{t}$.

  The vectors of the set $B\subset\NN^n$ produced by the Bellman--Ford
  circuit $\bphi$ correspond not to (simple) paths but rather to
  \emph{walks} of length at most $n-1$ from $s$ to $t$.  Since a walk
  can traverse the same edge many times, some vectors in $B$ have
  entries much larger than $1$.  Still, by \cref{prop:deg-upper}(1),
  $\deg{\bphi}=1$ holds: every (simple) $s$-$t$ path is also a walk of
  length at most $n-1$, implying that $A\subseteq B$.
\end{ex}

\subsection{Semantic versus syntactic degree}

The standard, ``syntactic'' definition of the degree is the following.
Each input node holding a variable has degree $1$. The degree of an OR
gate is the \emph{maximum} of the degrees of its input gates, and the
degree of an AND gate is the \emph{sum} of the degrees of its input
gates.  The following proposition shows that the semantic degree never
exceeds the syntactic degree.

\begin{prop}\label{prop:degree}
  Let $\bphi_1$ and $\bphi_2$ be any two monotone boolean
  circuits. Then 
  \[
  \mbox{$\deg{\bphi_1\lor \bphi_2} \leq
  \max\left\{\deg{\bphi_1}, \deg{\bphi_2}\right\}$ and
  $\deg{\bphi_1\land \bphi_2} \leq \deg{\bphi_1}+ \deg{\bphi_2}$.}
  \]
\end{prop}
\begin{proof}
  For $i\in\{1,2\}$, let $A_i\subseteq\{0,1\}^n$ be the set of
  minterms of the boolean function computed by $\bphi_i$, and let
  $B_i\subset\NN^n$ be the set of vectors produced by $\bphi_i$. Let
  $r_i=\deg{\bphi_i}$ be the semantic degree of~$\bphi_i$.

  Take an arbitrary minterm $a$ of $\bphi$.  If $\bphi=\bphi_1\lor
  \bphi_2$, then $B=B_1\cup B_2$ is the union of the set $B_1$ and
  $B_2$, and $a\in A_i$ for some $i\in\{1,2\}$. We know that $c\leq
  r_i \cdot a$ must hold for some vector $c$ in $\conv{B_i}\subseteq
  \conv{B}$. So, $\deg{\bphi_1\lor \bphi_2}\leq \max\{r_1, r_2\}$ in
  this case.

  If $\bphi=\bphi_1\land \bphi_2$, then $B=B_1+B_2$ is the Minkowski
  sum of the sets $B_1$ and $B_2$, and $a=\vvec{a}{1}\lor \vvec{a}{2}$
  is a componentwise OR of some minterms $\vvec{a}{1}\in A_1$ and
  $\vvec{a}{2}\in A_2$.  We know that $\vvec{c}{1}\leq r_1\cdot
  \vvec{a}{1}$ and $\vvec{c}{2}\leq r_i\cdot \vvec{a}{2}$ must hold
  for some vectors $\vvec{c}{1}\in\conv{B_1}$ and
  $\vvec{c}{2}\in\conv{B_2}$.  A well-known property of Minkowski sums
  is that $\conv{B_1}+\conv{B_2}=\conv{B_1+B_2}$ always holds. Hence,
  the vector $c=\vvec{c}{1}+\vvec{c}{2}$ belongs to $\conv{B}$ and
  satisfies $c=\vvec{c}{1}+\vvec{c}{2}\leq r_1\cdot \vvec{a}{1} +
  r_2\cdot \vvec{a}{2}\leq r_1\cdot a + r_2\cdot a =(r_1+r_2)\cdot a$.
  So, $\deg{\bphi_1\land \bphi_2}\leq r_1+r_2$ holds in this case.
\end{proof}

The following example illustrates that, together with
\cref{prop:degree}, the upper bound $\Min{A}{r}\leq \BBool{A}{r}$
given by \cref{thm:minexact} allows one to show that some minimization
problems \emph{can} be approximated by small $(\min,+)$ circuits
within (large) but finite factors by proving upper bounds for monotone
\emph{boolean} circuits of bounded semantic degree. Recall that
\emph{some} minimization problems cannot be approximated by $(\min,+)$
circuits of polynomial size within \emph{any} finite factor $r=r(n)$
at all (\cref{sec:minplus-explicit}).

\begin{ex}[Spanning trees]\label{ex:ST}
  In the \emph{minimum weight spanning tree} problem $\spanntree_n$,
  we are given an assignment of nonnegative real weights to the edges
  of $K_n$, and the goal is to compute the minimum weight of a
  spanning tree in $K_n$; the weight of a subgraph is the sum of
  weights of its edges. We have shown in~\cite{JS19} that
  $\Min{\spanntree_n}{1}=2^{\Omega(\sqrt{n})}$.

  On the other hand, the decision version of this problem is the graph
  connectivity problem.  Using the (pure) DP algorithm of Bellman and
  Ford, for every pair $(s,t)$ of vertices, the $s$-$t$ connectivity
  problem can be solved by a monotone boolean circuit $\bphi_{s,t}$ of
  size $O(n^3)$ and semantic degree $\deg{\bphi_{s,t}}=1$
  (see~\cref{ex:BF}).  So, the connectivity problem can be solved by
  the circuit $\bphi_{1,2}\land \bphi_{1,3}\land\cdots\land
  \bphi_{1,n}$ of size $O(n^4)$. By \cref{prop:degree}, the circuit
  has semantic degree $r\leq n-1$.  \Cref{thm:minexact} implies that
  $\Min{\spanntree_n}{r}=O(n^4)$ holds for some \emph{finite}
  factor~$r\leq n-1$.
\end{ex}

\section{Conclusion and open problems}
\label{sec:concl}
Developing a workable taxonomy of existing algorithmic paradigms in
\emph{rigorous} mathematical terms is an important long-term goal.
When pursuing this goal, the main difficulty is to prove unconditional
\emph{lower bounds} on the complexity of algorithms from particular
classes, that is, to prove lower bounds not relying on unproven
complexity assumptions like {\bf P} $\neq$ {\bf NP}.

In this paper, we consider the class of all pure DP algorithms, take
tropical circuits as their natural mathematical model, and prove the
first non-trivial (even super-polynomial) unconditional lower bounds
for \emph{approximating} pure DP algorithms in this model.  The
results imply that the approximation powers of greedy and pure DP
algorithms are incomparable.  Some interesting questions still remain
open.

Given a family $\f\subseteq 2^{[n]}$ of feasible solutions, and an
approximation factor $r\geq 1$ let, as before, $\Max{\f}{r}$ denote
the minimum number of gates in a $(\max,+)$ circuit approximating the
maximization problem $f(x)=\max_{S\in \f}\sum_{i\in S}x_i$ on $\f$
within the factor~$r$.  In the case of minimization problems and
$(\min,+)$ circuits, the corresponding complexity measure is $\Min{\f}{r}$.

\subsection{Minimization}
We have shown in \cref{thm:min-gap} that there \emph{exist} a lot of
monotone boolean functions $f$ such that minterms of $f$ are bases of
a matroid, and $f$ requires monotone boolean circuits of exponential
size. But we do not know of any \emph{explicit} matroid for which the
corresponding boolean function requires large monotone boolean
circuits.

\begin{problem}\label{prob:matroid}
  Prove a super-polynomial lower bound on the monotone boolean circuit
  complexity of an \emph{explicit} boolean function whose minterms are
  bases of a matroid.
\end{problem}

Let $\spanntree_n$ be the family of all spanning trees in a complete
$n$-vertex graph $K_n$. Since $\spanntree_n$ is a matroid, both
minimization and maximization problems can be solved exactly (within
factor $r=1$) by the greedy algorithm. On the other hand, we know that
$\Min{\spanntree_n}{1}=2^{\Omega(\sqrt{n})}$~\cite{JS19}. We also know
that $\Min{\spanntree_n}{r}=O(n^4)$ holds if factor $r=n-1$ is allowed
(\cref{ex:ST}).

\begin{problem}
  Is $\Min{\spanntree_n}{2}$ polynomial in $n$?
\end{problem}

\subsection{Maximization}
The next question concerns the \emph{maximization} problem on the
matroid $\spanntree_n$ of spanning trees in $K_n$. We know that, for
factor $r=1$, we have
$\Max{\spanntree_n}{1}=2^{\Omega(\sqrt{n})}$~\cite{JS19}.
\begin{problem}
  Is $\Max{\spanntree_n}{2}$ polynomial in $n$?
\end{problem}

In \cref{thm:matching}, we considered the maximum weight problem on
$k$-partite \emph{hypergraphs}.  For $k=2$, the calculations made in
the proof of \cref{thm:matching} result in a trivial bound.  This
rises a natural question: does a similar lower bound hold also for
matchings in \emph{bipartite} graphs? Let $\matching_n$ be the family
of all perfect matchings in a complete bipartite $n\times n$
graph. The greedy algorithm can approximate the maximization problem
on $\matching_n$ within the factor~$2$.

\begin{problem}
  Is $\Max{\matching_n}{2}$ polynomial in $n$?
\end{problem}

We have shown in \cref{thm:min-gap} that the \emph{minimization}
problem on many matroids cannot be efficiently approximated by pure DP
algorithms within any finite factor~$r$.  But what happens with
\emph{maximization} problems? By \cref{prop:counting1}, we know that
there are a lot of matroids $\f\subseteq 2^{[n]}$ such that
$\Max{\f}{1}=2^{\Omega(n)}$ but $\Max{\f}{r}\leq n^2$ holds already
for $r=1+o(1)$.

\begin{problem}
  Are there matroids, on which the maximization problem cannot be
  efficiently approximated by pure DP algorithms within some factor
  $r\geq 1+\epsilon $ for a constant $\epsilon >0$?
\end{problem}
Note that here we only ask for the mere \emph{existence}.  By
\cref{prop:counting1}, the answer is ``yes'' for $r=1$. But this
proposition and \cref{prop:counting2} indicate that direct counting
arguments may fail to answer this question for slightly larger
approximation factors~$r$.

\subsection{Tradeoffs between minimization and maximization}
If a family $\f$ of feasible solutions is \emph{uniform} (all sets of
$\f$ have the same cardinality), then $\Min{\f}{1}=\Max{\f}{1}$ (see,
for example, \cite[Lemma~2]{juk-SIDMA}).  That is, if we consider
\emph{exactly} solving tropical circuits (factor $r=1$), then there is
no difference between the tropical circuit complexity of the
minimization and the maximization problem on the \emph{same} (uniform)
set $\f$ of feasible solutions.

But the situation is entirely different if we consider
\emph{approximating} circuits: \cref{thm:min-gap,prop:counting1} give
us doubly-exponentially many in $n$ matroids $\f\subseteq 2^{[n]}$
such that $\Max{\f}{1+o(1)}\leq n^2$, but $\Min{\f}{r}=2^{\Omega(n)}$
for any finite factor $r=r(n)\geq 1$.

\begin{problem}
  Are there uniform families $\f$ for which the gap
  $\Max{\f}{r}/\Min{\f}{s}$ is exponential for $r\geq s>1$?
\end{problem}

Note that the separating family $\f$ is here required to be uniform
(or at least to form an antichain): without this requirement, the gap
can be \emph{artificially} made large. To see this, take an arbitrary
uniform family $\h\subseteq 2^{[n]}$ with large $\Max{\h}{r}$ (as in
Theorems~\ref{thm:hierarchy} and \ref{thm:matching}), and extend it to
a nonuniform family $\f$ by adding all single element sets. Then
$\Min{\f}{1}\leq n$ (just compute the minimum weight of a single
element), but $\Max{\f}{r}$ still remains large.

\subsection{Pure DP algorithms with subtraction}
\label{sec:subtraction}
Can the size of tropical approximating circuits be substantially
reduced by allowing (besides $\min/\max$ and $+$) also
\emph{subtraction} $(-)$ gates? In the case of the approximation
factor $r=1$ (exact solution), we already know the answer: subtraction
gates \emph{can} then even exponentially decrease the circuit
size. Namely, we already know that both directed and undirected
versions of the MST problem (minimum weight spanning tree problem) on
$n$-vertex graphs require tropical $(\min,+)$ circuits of size
$2^{\Omega(\sqrt{n})}$~\cite{jerrum,JS19} but, as shown by Fomin,
Grigoriev and Koshevoy~\cite{FGK}, both these problems are solvable by
tropical $(\min,+,-)$ circuits of size only $O(n^3)$.  Unfortunately,
no non-trivial lower bounds for $(\min,+,-)$ circuits are known so
far. So, at least two natural questions arise.
\begin{enumerate}
\item[$\circ$] Prove lower bounds for $(\min,+,-)$ circuits, at least
  when $r=1$.
\item[$\circ$] What about larger approximation factors $r>1$?
\end{enumerate}
Note that, when restricted to the \emph{boolean} domain $\{0,1\}$,
$(\min,+,-)$ circuits have the entire power of \emph{unrestricted}
boolean $(\lor,\land,\neg)$ circuits: $x\land y=\min(x,y)$, $x\lor
y=\min(1,x+y)$ and $\neg x= 1-x$. The point, however, is that
$(\min,+,-)$ circuits must correctly work over the entire \emph{real}
domain~$\RR_+$.

\appendix

\section{Greedy algorithms}
\label{app:greedy}
Since we compared the approximation power of tropical circuits (and
pure DP algorithms) with that of the greedy algorithm, here we specify
what we actually mean by ``the'' greedy algorithm.

Let $\f\subseteq 2^E$ be some family of feasible solutions forming an
antichain (no two members of $\f$ are comparable under set inclusion).
Given an ordering $e_1,\ldots,e_n$ of the elements of $E$, there are
two trivial heuristics to end up with a member of~$\f$ by treating the
elements one-by-one in this fixed order.
\begin{description}
\item[\mbox{\it First-in}] Start with the empty partial solution,
  treat the elements one-by-one and, at each step, \emph{add} the next
  element to the current partial solution if and only if the extended
  partial solution still lies in at least one feasible solution.

\item[\mbox{\it First-out}] Start with the entire set $E$ as a
  partial solution, treat the elements one-by-one and, at each step,
  \emph{remove} the next element from the current partial solution if
  and only if the reduced partial solution still contains at least one
  feasible solution.
\end{description}

Recall that an optimization (maximization or minimization) problem on
$\f$ is, given an assignment of nonnegative real weights to the ground
elements, to compute the maximum or the minimum weight of a feasible
solution, the latter being the sum of weights of its elements.

In this paper, by \emph{the greedy algorithm} we always mean the
algorithm which, on every input weighting $x:E\to\RR_+$, starts with
the heaviest-first ordering $x(e_1)\geq x(e_2)\geq \ldots\geq x(e_n)$
of the elements of $E$, and uses:
\begin{itemize}
\item[-] the first-in heuristic (``best-in'' strategy) in the case of
  \emph{maximization};

\item[-] the first-out heuristic (``worst-out'' strategy) in the case
  of \emph{minimization}.
\end{itemize}
That is, at each step, the ``oracle'' of the maximizing greedy
algorithm decides whether the current set is still contained in at
least one feasible solution, while that of the minimizing greedy
algorithm decides whether the current set still contains at least one
feasible solution.

We denote the approximation factor achieved by the greedy algorithm on
a corresponding optimization (minimization or maximization) problem on
$\f$ by~$\fgreedy{\f}$.  It is well known (see, for example,
\cite[Theorem~1.8.4]{oxley} that $\fgreedy{\f}=1$ if and only if $\f$
is (the family of bases of) a matroid.  If $\f$ is not a matroid, then
greedy algorithms can only approximate the corresponding optimization
problems. In this case, it is already crucial what greedy strategy is
used.

\begin{ex}\label{ex:bad}
  The choice of these special heuristics (first-in for maximization
  and first-out for minimization) is not an accident. Namely, a greedy
  algorithm starting with the lightest-first ordering $x(e_1)\leq
  x(e_2)\leq \ldots\leq x(e_n)$, and using the first-out heuristic
  (``worst-out'' strategy) for maximization or first-in heuristic
  (``best-in'' strategy) for minimization would be unable to
  approximate some optimization problems within any finite factor.  To
  give a simple example, consider the path with three nodes
  $\scalebox{0.8}{\xygraph { !{<0cm,0cm>;<1cm,0cm>:<0cm,-0.9cm>::}
      !{(0,0)}*+=[o]{\bullet}="a" !{(1,0)}*+=[o]{\bullet}="b"
      !{(2,0)}*+=[o]{\bullet}="c" !{(0,-0.3)}*+{a}="a1"
      !{(1,-0.3)}*+{b}="b1" !{(2,-0.3)}*+{c}="c1" "a"-"b" "b"-"c" }}$,
  and let $\f$ be the family consisting of just two sets $\{a,c\}$ and
  $\{b\}$ (the maximal independent sets in this path).  If we take an arbitrarily large number $M>1$, and give
  weights $x(a)=0$, $x(b)=1$ and $x(c)=M$, then both these greedy algorithms will treat the vertices in the order $a,b,c$. The worst-out maximizing greedy on $\f$ will output $x(b)=1$ while the optimum is $M$, and the best-in greedy for
  minimization will output $x(a)+x(c)=0+M$, while the optimum is $1$. In both
  cases, the achieved approximation factor is $r\geq M$ (unbounded).
\end{ex}

If, however, the greedy algorithm uses the ``right'' strategies for
maximization and for minimization, then the approximation factor is
always \emph{bounded} (albeit possibly growing with the size of
feasible solutions). Say that a family $\f$ of sets is
$m$-\emph{bounded} if $|S|\leq m$ holds for all $S\in\f$.

\begin{prop}\label{prop:greedy}
  For every $m$-bounded family $\f$, we have $\fgreedy{\f}\leq m$, and
  there exist $m$-bounded antichains $\f$ for which $\fgreedy{\f}=m$.
\end{prop}

\begin{proof}
  To show the upper bound, take an arbitrary weighting
  $x:E\to\RR_+$. Consider the heaviest-first ordering $x(e_1)\geq
  \ldots \geq x(e_i)\geq \ldots\geq x(e_n)$.  Let $e_i$ be the
  \emph{first} element accepted by the greedy algorithm.  Let $S\in\f$
  be an optimal solution for the input $x$, and $A\in\f$ be the
  solution found by the algorithm. Let also $x(S)=\sum_{i\in S}
  x(e_i)$ and $x(A)=\sum_{i\in A} x(e_i)$ be their weights.

  If this is the maximizing (best-in) greedy, then $e_i$ is the first
  element belonging to at least one feasible set. So, $S\cap
  \{e_1,\ldots,e_{i-1}\}=\emptyset$, implying that $x(S)\leq |S|\cdot
  x(e_i)\leq m\cdot x(e_i)\leq m\cdot x(A)$, as desired.

  If this is the minimizing (worst-out) greedy, then
  $\{e_{i+1},\ldots,e_n\}$ cannot contain any feasible solution (for
  otherwise, $e_i$ would not be accepted). So, some element $e_j$ with
  $j\leq i$ must belong to the optimal solution $S$. But then
  $x(S)\geq x(e_j)\geq x(e_i)$, whereas $x(A)\leq |A|\cdot x(e_i)\leq
  m\cdot x(e_i)$, implying that $x(A)\leq m\cdot x(S)$, as desired.

  To show that $\fgreedy{\f}\geq m$ holds for some $m$-bounded
  antichains $\f$, take an arbitrarily small number $\epsilon>0$, and
  consider the star $K_{1,m}$ centered in $a$ and with leaves
  $b_1,\ldots,b_m$. Let $\f$ consist of the only two maximal
  independent sets $\{a\}$ and $\{b_1,\ldots,b_m\}$ in this graph.
  Give the weight $\scalar:=1/(1-\err/2)>1$ to the center $1$, and
  weights $1$ to the leaves. The maximizing (best-in) greedy will
  output $\scalar$ while the optimum is $m$, and the minimizing
  (worst-out) greedy algorithms will output $m$ while the optimum is
  $\scalar$. In both cases, the achieved approximation factor is
  $r\geq m/\scalar > (1-\epsilon)m$.
\end{proof}

\section{Sidon sets: proof of Theorem~\ref{thm:sidon}}
\label{app:sidon}
A set $A\subset\NN^n$ of vectors is a \emph{Sidon set} if for all
vectors $a,b,c,d\in A$: if $a+b=c+d$, then $\{a,b\}=\{c,d\}$. That is,
knowing the sum of two vectors in $A$, we know which vectors were
added.  Let (as before) $\Max{A}{r}$ denote the minimum size of a
tropical $(\max,+)$ circuit $r$-approximating the problem
$f(x)=\max_{a\in A}\skal{a,x}$ on~$A$.

Let $m$ be an odd integer, and $n=4m$. Our goal is to show that then
there is an explicit Sidon set $A\subseteq\{0,1\}^n$ such that
$\Max{A}{1}\geq 2^{n/4}$ but $\Max{A}{2}\leq n$.  For this, consider
the cubic parabola $C=\{(z,z^3)\colon z\in \{0,1\}^m\}\subseteq
\gf{2^{2m}}$.  As customary, we view vectors in $z\in\{0,1\}^m$ as
coefficient-vectors of polynomials of degree at most $m-1$ over
$\gf{2}$ when rising them to a power. Note, however, that in the
definition of Sidon sets, the sum of vectors is taken over the
semigroup $(\NN,+)$, not over $\gf{2}$; in particular, $a+a=0$ holds
only for $a=0$.

For a finite set $A\subset\NN^n$ of vectors, let $\Un{A}$ denote the
minimum size of a Minkowski $(\cup,+)$ circuit producing~$A$.  We will
use the following three facts. Recall that a set $A\subseteq\{0,1\}^n$
is \emph{uniform} if all its vectors have the same number of ones.

\begin{itemize}
\item[(1)] The cubic parabola $C\subseteq \{0,1\}^{2m}$ is a Sidon
  set~\cite[Theorem~2]{lindstrom}.

\item[(2)] $\Un{A}\geq |A|$ holds for every Sidon set
  $A\subset\NN^n$~\cite[Theorem~1]{GS12}.

\item[(3)] If $A\subseteq\{0,1\}^n$ is uniform, then $\Max{A}{1}\geq
  \Un{A}$~\cite[Theorem~2.9]{jerrum}.
\end{itemize}
The cubic parabola $C$ is not uniform, and we cannot apply (3) to it.
But, using a simple trick (suggested by Igor Sergeev, personal
communication), we can extend this set to a uniform Sidon set.  For a
$0$-$1$ vector $a$, let $\bar{a}$ denote the componentwise negation of
$a$. For example, if $a=(0,0,1)$ then $\bar{a}=(1,1,0)$. Consider the
following set of vectors:
\[
A= \{(c,\bar{c})\colon c\in C\}=\left\{(a,a^3,\bar{a},\bar{a^3})\colon
  a\in\{0,1\}^m\right\}\subseteq\{0,1\}^{n}\,.
\]
This set is already uniform: every vector of $A$ has exactly $2m$
ones. The set $A$ is also a Sidon set because, by (1), the set $C$ was
such.  So, (2) and (3) imply that $\Max{A}{1}\geq |A|=2^m=2^{n/4}$.

It remains therefore to prove the upper bound $\Max{A}{2}\leq n$.  We
have $n=4m$ variables $x_1,\ldots,x_{4m}$. Our approximating circuit
will solve the maximization problem on the set $B=B'\cup B''$, where
\[
B'=\{(a,0,\bar{a},0)\colon a\in \{0,1\}^m\}\ \ \mbox{ and }\ \
B''=\{(0,a,0,\bar{a})\colon a\in \{0,1\}^m\}\,.
\]
The maximization problem on $B$ is to compute $f(x)=\max\{g(x),
h(x)\}$, where
\begin{align*}
  g(x)&=\max\ \sum_{i=1}^m a_ix_i+\sum_{i=2m+1}^{3m} (1-a_{i})x_{i}\,;\\
  h(x)&=\max\ \sum_{i=m+1}^{2m} a_ix_i+\sum_{i=3m+1}^{4m}
  (1-a_{i})x_{i}
\end{align*}
with both maximums taken over all vectors $a\in\{0,1\}^{4m}$.  Since
$g(x)$ is just the sum
$\max\{x_1,x_{2m+1}\}+\max\{x_2,x_{2m+2}\}+\cdots
+\max\{x_m,x_{3m}\}\,, $ and similarly for $h(x)$, the maximization
problem $f$ can be solved using only $4m=n$ gates.

It remains to show that $f$ indeed approximates the maximization
problem on $A$ within factor $r=2$. As we have shown in
\cref{sec:exact} (see \cref{lem:max}), this happens precisely when the
set $B$ lies below $A$, and the set $\tfrac{1}{2}\cdot A$ lies below
the convex hull $\conv{B}$ of $B$. It is clear that the first subset
$B'$ of $B$ lies below $A$.  We have to show that this holds also for
the second subset $B''$. For this, it is enough to show that $B''$
coincides with the set of all vectors $(0,a^3,0,\bar{a^3})$ for
$a\in\{0,1\}^m$.

It is known that a polynomial $x^k$ permutes $\gf{q}$ if and only if
$q-1$ and $k$ are relatively prime; see, for example, Lidl and
Niederreiter \cite[Theorem~7.8]{LiedlN86}. In our case, we have
$q=2^m$ and $k=3$. Since $m$ is odd, we have $m=2t+1$ for some
$t\in\NN$. Easy induction on $t$ shows that $p(t):=2^{2t+1}+1$ is
divisible by $3$: the basis $t=0$ is obvious, because $p(0)=3$, and
the induction step $p(t+1)=2^{2(t+1)+1}+1= 4(2^{2t+1}+1)-3=4\cdot
p(t)-3$ follows from the induction hypothesis.  So, $q-1=p(t)-2$
cannot be divisible by $3$, that is, $q-1$ and $3$ are relatively
prime and, hence, the mapping $a\mapsto a^3$ is a
\emph{bijection}. This gives us a crucial fact:
\[
\left\{(0,a^3,0,\bar{a^3})\colon a\in \{0,1\}^m\right\}
=\left\{(0,a,0,\bar{a})\colon a\in \{0,1\}^m\right\} =B''\,.
\]
Hence, the entire set $B=B'\cup B''$ lies below $A$, that is, every
vector of $B$ is covered by at least one vector of $A$.  By
\cref{lem:max}, it remains to show that the set $\tfrac{1}{2}\cdot A$
lies below the convex hull $\conv{B}$.  So, take an arbitrary vector
$u=\tfrac{1}{2}\cdot(a,a^3,\bar{a},\bar{a^3})$ in $\tfrac{1}{2}\cdot
A$.  This vector is a convex combination $\tfrac{1}{2}\cdot
v+\tfrac{1}{2}\cdot w$ of vectors $v=(a,0,\bar{a},0)$ and
$w=(0,a^3,0,\bar{a^3})$ of $B$, as desired.  \qed

\section*{Acknowledgments}
We thank Georg Schnitger and Igor Sergeev for inspiring discussions.

\bibliographystyle{siamplain}

\end{document}